\documentclass[11pt,a4paper]{article}

\usepackage[inline]{enumitem}
\usepackage{setspace}
\usepackage{amssymb,amsthm,amsmath,mathtools}
\usepackage{thmtools,thm-restate}
\usepackage[round]{natbib}
\usepackage{pifont}

\usepackage[usenames,svgnames,xcdraw,table]{xcolor}
\definecolor{DarkGreen}{rgb}{0.1,0.5,0.1}
\usepackage[backref=page]{hyperref}
\hypersetup{
	colorlinks=true,
	linkcolor=red,%color of internal links
	urlcolor=DarkGreen,%color of external links
	citecolor=blue%color of links to bibliography
}
\renewcommand*{\backref}[1]{}
\renewcommand*{\backrefalt}[4]{
    \ifcase #1 (Not cited.)
    \or        (Cited on page~#2)
    \else      (Cited on pages~#2)
    \fi}
\usepackage[capitalise,noabbrev]{cleveref}
\Crefname{property}{Property}{Properties}
\Crefname{example}{Example}{Examples}
\Crefname{table}{Table}{Tables}
\usepackage{nicefrac}
\usepackage{enumerate}
\usepackage{etoolbox}
\usepackage[margin=1in]{geometry}
\usepackage[T1]{fontenc}
\usepackage[tt=false]{libertine}
\usepackage{url}
\usepackage[utf8]{inputenc}
\usepackage[small]{caption}
\usepackage{booktabs}
\usepackage{mathrsfs,bm,amsfonts,dsfont}
\usepackage{authblk}
\usepackage{array,multirow,graphicx,bigdelim}
\newcolumntype{P}[1]{>{\centering\arraybackslash\hspace{0pt}}p{#1}}
\usepackage{float}

\makeatletter
\renewcommand{\paragraph}{
  \@startsection{paragraph}{4}
  {\z@}{1.0ex \@plus 1ex \@minus .2ex}{-1em}
  {\normalfont\normalsize\bfseries}
}
\let\oldnl\nl
\newcommand{\nonl}{\renewcommand{\nl}{\let\nl\oldnl}}
\makeatother

\usepackage[labelfont={normalfont,bf},textfont=it]{caption}
\usepackage{subcaption}
\usepackage{upgreek}

\newtheorem{definition}{Definition}
\newtheorem{remark}{Remark}

\newtheorem{prop}{Proposition}
\crefalias{prop}{proposition}
\usepackage{relsize}
\usepackage{mdframed}

\newcommand{\A}{\mathcal A}
\renewcommand{\>}{\succ}
\newcommand{\aff}{\mathrm{\textup{aff}}}
\newcommand{\cost}{\textup{\texttt{time}}}

\newcommand{\D}{\mathcal{D}}

\newcommand{\dkt}{\mathrm{d_\textup{kt}}}

\newcommand{\E}{\mathbb{E}}
\newcommand{\Ind}{\mathbf{1}}

\newcommand{\Kem}{\textsc{Kemeny}}
\newcommand{\kPL}{\textup{$k$-PL}}
\newcommand{\kMM}{\textup{$k$-MM}}

\newcommand{\ml}{\mathcal L}

\newcommand{\lin}{\mathrm{\textup{lin}}}

\newcommand{\MM}{\textup{MM}}
\newcommand{\mytheta}{\uptheta}
\newcommand{\N}{\mathbb{N}}

\newcommand{\NP}{\textup{NP}}
\newcommand{\NPc}{\textup{NP-c}}
\newcommand{\NPC}{\textup{NP-complete}}
\newcommand{\NPH}{\textup{NP-hard}}
\newcommand{\Otilde}{\tilde{\mathcal{O}}}
\renewcommand{\P}{\textup{Poly}}
\newcommand{\PL}{\textup{PL}}
\newcommand{\poly}{\textup{poly}}

\newcommand{\PTAS}{\textup{PTAS}}
\newcommand{\Q}{\mathbb{Q}}
\newcommand{\Rec}{$(\D,\w)$\textsc{-Recommendation}}
\newcommand{\rank}{\textup{rank}}

\newcommand{\si}{\text{\#moves}}

\newcommand{\Unif}{\textup{Unif}}
\newcommand{\w}{\mathbf{w}}
\newcommand{\waff}{\w_\aff}
\newcommand{\wlin}{\w_\lin}

\newcommand{\WFAST}{\textsc{Weighted Feedback Arc Set in Tournaments}}
\newcommand{\wfast}{\textsc{WFAST}}
\newcommand{\Z}{\mathbb{Z}}
\newcommand{\R}{\mathbb{R}}

\title{Minimizing Time-to-Rank:\\[0.2em]\smaller{} A Learning and Recommendation Approach}

\author[a]{Haoming Li}
\author[b]{Sujoy Sikdar}
\author[c]{Rohit Vaish}
\author[d]{Junming Wang}
\author[e]{Lirong Xia}
\author[f]{Chaonan Ye}
\affil[a]{Duke University\\
	{\small\texttt{haoming.li@duke.edu}}}
\affil[b]{Rensselaer Polytechnic Institute\\
	{\small\texttt{sikdas@rpi.edu}}}
\affil[c]{Rensselaer Polytechnic Institute\\ 
	{\small\texttt{vaishr2@rpi.edu}}}
\affil[d]{Rensselaer Polytechnic Institute\\ 
	{\small\texttt{wangj33@rpi.edu}}}
\affil[e]{Rensselaer Polytechnic Institute\\
	{\small\texttt{xial@cs.rpi.edu}}}
\affil[f]{Stanford University\\ 
	{\small\texttt{canonyeee@gmail.com}}}

\begin{document}

\maketitle

\begin{abstract}
Consider the following problem faced by an online voting platform: A user is provided with a list of alternatives, and is asked to rank them in order of preference using only drag-and-drop operations. The platform's goal is to recommend an initial ranking that minimizes the time spent by the user in arriving at her desired ranking. We develop the first optimization framework to address this problem, and make theoretical as well as practical contributions. On the practical side, our experiments on Amazon Mechanical Turk provide two interesting insights about user behavior: First, that users' ranking strategies closely resemble \emph{selection} or \emph{insertion} sort, and second, that the time taken for a drag-and-drop operation depends \emph{linearly} on the number of positions moved. These insights directly motivate our theoretical model of the optimization problem. We show that computing an optimal recommendation is \NPH{}, and provide exact and approximation algorithms for a variety of special cases of the problem. Experimental evaluation on MTurk shows that, compared to a random recommendation strategy, the proposed approach reduces the (average) time-to-rank by up to $50\%$.
\end{abstract}

\section{Introduction}
\label{sec:Introduction}

Eliciting preferences in the form of rankings over a set of alternatives is a common task in social choice, crowdsourcing, and in daily life. For example, the organizer of a meeting might ask the participants to rank a set of time-slots based on their individual schedules. Likewise, in an election, voters might be required to rank a set of candidates in order of preference.

Over the years, computerized systems have been increasingly used in carrying out preference elicitation tasks such as the ones mentioned above. Indeed, recently there has been a proliferation of online voting platforms such as CIVS, OPRA, Pnyx, RoboVote, and Whale$^4$.\footnote{CIVS (\url{https://civs.cs.cornell.edu/}), OPRA (\url{opra.io}), Pnyx (\url{https://pnyx.dss.in.tum.de/}), RoboVote (\url{http://robovote.org/}), Whale$^4$(\url{https://whale.imag.fr/}).} In many of these platforms, a user is presented with an arbitrarily ordered list of alternatives, and is asked to shuffle them around in-place using \emph{drag-and-drop} operations until her desired preference ordering is achieved. \Cref{fig:ui} illustrates the use of drag-and-drop operations in sorting a given list of numbers.

\begin{figure}[tp]
	\centering
	\includegraphics[width=0.8\linewidth]{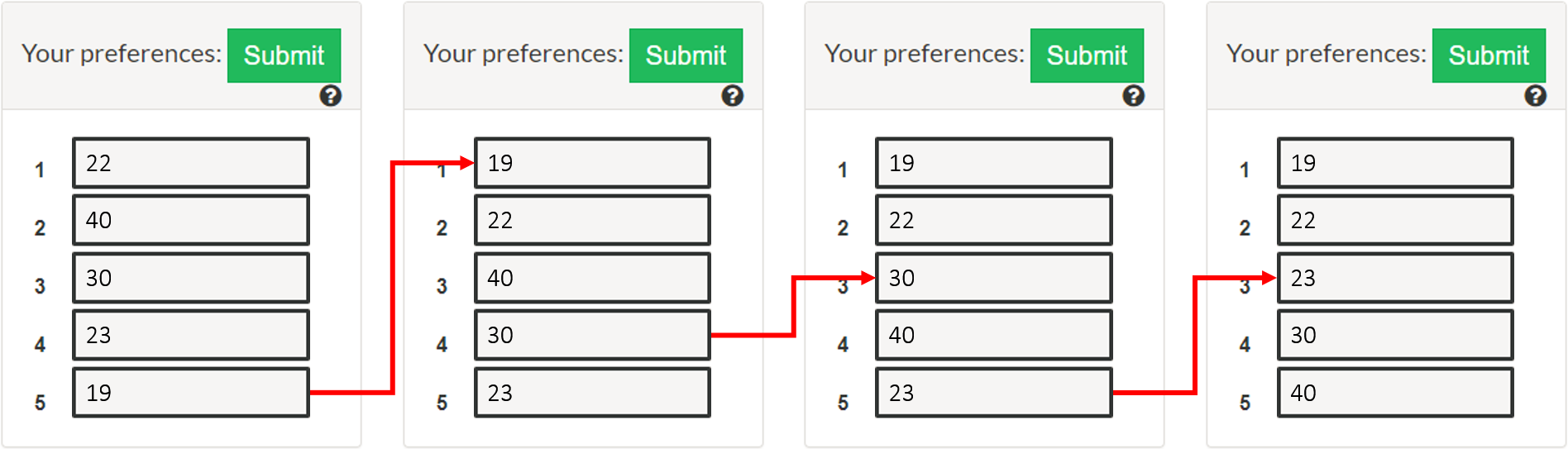}
	\caption{\label{fig:ui} \small Sorting via drag-and-drop operations.}
\end{figure}

Our focus in this work is on \emph{time-to-rank}, i.e., the time it takes for a user to arrive at her desired ranking, starting from a ranking suggested by the platform and using only drag-and-drop operations. We study this problem from the perspective of the voting platform that wants to \emph{recommend} an optimal initial ranking to the user (i.e., one that \emph{minimizes} time-to-rank). Time to accomplish a designated task is widely considered as a key consideration in the usability of automated systems~\citep{Bevan15:ISO,Albert13:Measuring}, and serves as a proxy for user effort. Indeed, `time on task' was identified as a key factor in the usability and efficiency of computerized voting systems in a 2004 report by NIST to the U.S. Congress for the Help America Vote Act (HAVA)~\citep{Laskowski04:Improving}. In crowdsourcing, too, {time on task} plays a key role in the recruitment of workers, quality of worker participation, and in determining payments~\citep{Cheng15:Measuring,Maddalena16:Crowdsourcing}.

Note that the initial ranking suggested by the platform can have a significant impact on the time spent by the user on the ranking task. Indeed, if the user's preferences are known beforehand, then the platform can simply recommended it to her and she will only need to verify that the ordering is correct. In practice, however, users' preferences are often \emph{unknown}. Furthermore, users employ a wide variety of \emph{ranking strategies}, and based on their proficiency with the interface, users can have very different \emph{drag-and-drop times}. All these factors make the task of predicting the time-to-rank and finding an optimal recommendation challenging and non-trivial.

We emphasize the subtle difference between our problem and that of \emph{preference elicitation}. The latter involves repeatedly asking questions to the users (e.g., in the form of pairwise comparisons between alternatives) to gather enough information about their preferences. By contrast, our problem involves a one-shot recommendation followed by a series of drag-and-drop operations by the user until her desired ranking is achieved. There is an extensive literature on preference eliciation~\citep{Conen01:Minimal,Conitzer02:Elicitation,Blum04:Preference,Boutilier13:Computational,Busa-Fekete2014:Preference-based,Azari13:Preference,Zhao2018:A-Cost-Effective}. Yet, somewhat surprisingly, the problem of recommending a ranking that minimizes users' time and effort has received little attention. Our work aims to address this gap.

%In particular, we are not aware of any work that models the drag-and-drop behavior of the users.

\paragraph{Our Contributions} We make contributions on three fronts:
\begin{itemize}
    \item On the \emph{conceptual} side, we propose the problem of minimizing time-to-rank and outline a framework for addressing it (\Cref{fig:high}).
    
    \item On the \emph{theoretical} side, we formulate the optimization problem of finding a recommendation to minimize time-to-rank (\Cref{sec:results}). We show that computing an optimal recommendation is \NPH{}, even under highly restricted settings (\Cref{thm:HardnessResults}). We complement the intractability results by providing a number of exact (\Cref{thm:ExactAlgorithms}) and approximation algorithms (\Cref{thm:RecPTAS,thm:BordaApproxAlgo,thm:ApproxGeneralWeights}) for special cases of the problem.
    
    \item We use \emph{experimental analysis} for the dual purpose of motivating our modeling assumptions as well as justifying the effectiveness of our approach (\Cref{sec:exp}). Our experiments on Amazon Mechanical Turk reveal two insights about user behavior (\Cref{subsec:User_Behavior}): (1) The ranking strategies of real-world users closely resemble \emph{insertion/selection sort}, and (2) the drag-and-drop time of an alternative varies \emph{linearly} with the distance moved. Additionally, we find that a simple adaptive strategy (based on the Borda count voting rule) can reduce time-to-rank by \emph{up to $50\%$} compared to a random recommendation strategy (\Cref{subsec:Expt_Evaluation}), validating the usefulness of the proposed framework.
\end{itemize}

\subsection{Overview of Our Framework}
\label{subsec:Overview_of_our_Framework}

\begin{figure}
	\centering
	\includegraphics[height=1.4in]{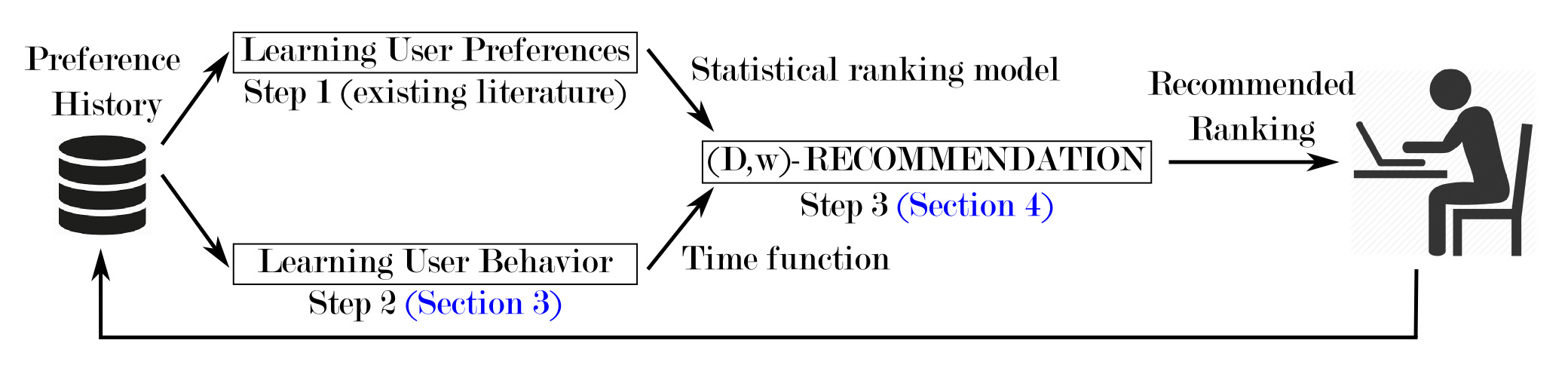}
	\caption{\label{fig:high} \small High level overview of our framework. Our technical contributions are highlighted in blue.}
	\vspace{-0.05in}
\end{figure}

\Cref{fig:high} illustrates the proposed framework which consists of three key steps. In Step 1, we \emph{learn user preferences} from historical data by developing a statistical ranking model, typically in the form of a distribution $\D$ over the space of all rankings (refer to \Cref{sec:Preliminaries} for examples of ranking models). In Step 2, which runs in parallel to Step 1, we \emph{learn user behavior}; in particular, we identify their \emph{sorting strategies} (\Cref{subsec:Sorting_Algorithms}) as well as their \emph{drag-and-drop times} (\Cref{subsec:Drag_and_drop}). Together, these two components define the \cost{} function which models the time taken by a user in transforming a given initial ranking $\sigma$ into a target ranking $\tau$, denoted by $\cost(\sigma,\tau)$. The ranking model $\D$ from Step 1 and the \cost{} function from Step 2 together define the recommendation problem in Step 3, called \Rec{} (the parameter $\w$ is closely related to the $\cost$ function; we elaborate on this below). This is the optimization problem of computing a ranking $\sigma$ that minimizes the expected time-to-rank of the user, i.e., minimizes $\E_{\tau \sim \D} [\cost(\sigma,\tau)]$. The user is then recommended $\sigma$, and her preference history is updated.

The literature on learning statistical ranking models is already well-developed \citep{Guiver09:Bayesian,Awasthi14:Learning,Lu2014:Effective,Zhao16:Learning}. Thus, while this is a key ingredient of our framework (Step 1), in this work we choose to focus on Steps 2 and 3, namely, learning user behavior and solving the recommendation problem.

Recall that the $\cost$ function defines the time taken by a user in transforming a given ranking $\sigma$ into a target ranking $\tau$. For a user who follows a fixed sorting algorithm (e.g., insertion or selection sort), the $\cost$ function can be broken down into \begin{enumerate*}[label=(\arabic*)] \item the number of drag-and-drop operations suggested by the sorting algorithm, and, \item the (average) time taken for each drag-and-drop operation by the user\end{enumerate*}. As we will show in \Cref{lem:EquivalenceOfSortingAlgorithms} in \Cref{subsec:Sorting_Algorithms}, point (1) above is \emph{independent} of the choice of the sorting algorithm. Therefore, the $\cost$ function can be equivalently defined in terms of the \emph{weight function} $\w$, which describes the time taken by a user, denoted by $\w(\ell)$, in moving an alternative by $\ell$ positions via a drag-and-drop operation. For this reason, we use $\w$ in the formulation of \Rec{}.

\paragraph{Applicability} Our framework is best suited for the users who have already formed their preferences, so that the recommended ranking does not bias their preferences. This is a natural assumption in some applications, such as in the meeting organization example in \Cref{sec:Introduction}. In general, however, it is possible that a user, who is undecided between options $A$ and $B$, might prefer $A$ over $B$ if presented in that order by the recommended ranking. A careful study of such biases (aka ``framing effect'') is an interesting direction for future work. 

\paragraph{Additional Related Work} 
Our work is related to the literature on inferring a ground truth ordering from noisy information~\citep{Braverman08:Noisy}, and aggregating preferences by minimizing some notion of distance to the observed rankings such as the total Kendall's Tau distance~\citep{Procaccia16:Optimal}. Previous work on preference learning and learning to rank can also be integrated in our framework~\citep{Liu11:Learning,Lu2014:Effective,Khetan2016:Data-driven,Agarwal2016:On-ranking,Negahban17:Rank,Zhao2018:Composite}. 

\section{Preliminaries}
\label{sec:Preliminaries}

Let $A = \{a_1,\dots,a_m\}$ denote a set of $m$ \emph{alternatives}, and let $\ml(A)$ be the set of all linear orders over $A$. For any $\sigma \in \ml(A)$, $a_i \>_\sigma a_j$ denotes that $a_i$ is preferred over $a_j$ under $\sigma$, and let $\sigma(k)$ denote the $k^{\text{th}}$ most preferred alternative in $\sigma$. A set of $n$ linear orders $\{ \sigma^{(1)},\dots,\sigma^{(n)} \}$ is called a \emph{preference profile}.

\begin{definition}[Kendall's Tau distance; \citealp{K38new}]
	\label{defn:KendallTau}
	Given two linear orders $\sigma, \sigma' \in \ml(A)$, the Kendall's Tau distance $\dkt(\sigma,\sigma')$ is the number of pairwise disagreements between $\sigma$ and $\sigma'$. That is,
	$\dkt(\sigma,\sigma') \coloneqq \sum_{a_i,a_j \in A} \Ind[a_j \succ_{\sigma'} a_i \text{ and } a_i \>_\sigma a_j]$, where $\Ind$ is the indicator function.
\end{definition}

\begin{definition}[Plackett-Luce model; \citealp{Plackett75:Analysis,Luce59:Individual}]
	\label{defn:PL}
	Let $\mytheta \coloneqq (\theta_1,\dots,\theta_m)$ be such that $\theta_i \in (0,1)$ for each $i \in [m]$ and $\sum_{i \in [m]} \theta_i = 1$. Let $\Theta$ denote the corresponding parameter space. The Plackett-Luce $(\PL{})$ model parameterized by $\mytheta \in \Theta$ defines a distribution over the set of linear orders $\ml(A)$ as follows: The probability of generating $\sigma \coloneqq (a_{i_1} \> a_{i_2} \> \dots \> a_{i_m})$ is given by
	\begin{align*}
	\Pr(\sigma | \mytheta) = \frac{\theta_{i_1}}{\sum_{\ell=1}^m \theta_{i_\ell}} \cdot \frac{\theta_{i_2}}{\sum_{\ell=2}^m \theta_{i_\ell}} \cdot \dots \cdot \frac{\theta_{i_{m-1}}}{\theta_{i_{m-1}} + \theta_{i_m}}.
	\end{align*}
	
	More generally, a $k$-mixture Plackett-Luce model $(\kPL{})$ is parameterized by $\{\gamma^{(\ell)},\mytheta^{(\ell)}\}_{\ell=1}^k$, where $\sum_{\ell=1}^k \gamma^{(\ell)} = 1$, $\gamma^{(\ell)} \geq 0$ for all $\ell \in [k]$, and $\mytheta^{(\ell)} \in \Theta$ for all $\ell \in [k]$. The probability of generating $\sigma \in \ml(A)$ is given by
	$\Pr(\sigma | \{\gamma^{(\ell)},\mytheta^{(\ell)}\}_{\ell=1}^k) = \sum_{\ell=1}^k \gamma^{(\ell)} \Pr(\sigma | \mytheta^{(\ell)}).$
\end{definition}

\begin{definition}[Mallows model; \citealp{Mallows57:Non-null}]
	The Mallows model $(\MM)$ is specified by a reference ranking $\sigma^* \in \ml(A)$ and a dispersion parameter $\phi \in (0,1]$. The probability of generating a ranking $\sigma$ is given by
	$\Pr(\sigma | \sigma^*,\phi) = \frac{\phi^{\dkt(\sigma,\sigma^*)}}{Z},$
	where $Z = \sum_{\sigma' \in \ml(A)} \phi^{\dkt(\sigma',\sigma^*)}$.
	
	More generally, a $k$-mixture Mallows model $(\kMM)$ is parameterized by $\{\gamma^{(\ell)},\sigma^*_{(\ell)},\phi_{(\ell)}\}_{\ell=1}^k$, where $\sum_{\ell=1}^k \gamma^{(\ell)} = 1$, $\gamma^{(\ell)} \geq 0$ for all $\ell \in [k]$, and $\sigma^*_{(\ell)} \in \ml(A), \phi_{(\ell)} \in (0,1]$ for all $\ell \in [k]$. The probability of generating $\sigma \in \ml(A)$ is given by
	$\Pr(\sigma | \{\gamma^{(\ell)},\sigma^*_{(\ell)},\phi_{(\ell)}\}_{\ell=1}^k) = \sum_{\ell=1}^k \gamma^{(\ell)} \Pr(\sigma | \sigma^*_{(\ell)},\phi_{(\ell)}).$
\end{definition}

\begin{definition}[Uniform distribution]
	Under the uniform distribution $(\Unif{})$ supported on a preference profile $\{\sigma^{(i)}\}_{i=1}^n$, the probability of generating $\sigma \in \ml(A)$ is $\frac{1}{n}$ if $\sigma \in \{\sigma^{(i)}\}_{i=1}^n$ and $0$ otherwise.
\end{definition}

\section{Modeling User Behavior}
\label{sec:ModelingTime}

In this section, we will model the time spent by the user in transforming the recommended ranking $\sigma$ into the target ranking $\tau$. Our formulation involves the \emph{sorting strategy} of the user (\Cref{subsec:Sorting_Algorithms}) as well as her \emph{drag-and-drop time} (\Cref{subsec:Drag_and_drop}).

\subsection{Sorting Algorithms}
\label{subsec:Sorting_Algorithms}
A \emph{sorting algorithm} takes as input a ranking $\sigma \in \ml(A)$ and performs a sequence of \emph{drag-and-drop} operations until the target ranking is achieved. At each step, an alternative is moved from its current position to another (possibly different) position and the current ranking is updated accordingly. Below we will describe two well-known examples of sorting algorithms: \emph{selection sort} and \emph{insertion sort}. Let $\sigma^{(k)}$ denote the \emph{current} list at time step $k \in \{1,2,\dots\}$ (i.e., \emph{before} the sorting operation at time step $k$ takes place). Thus, $\sigma^{(1)} = \sigma$. For any $\sigma \in \ml(A)$, define the \emph{$k$-prefix set} of $\sigma$ as $P_k(\sigma) \coloneqq \{\sigma(1),\sigma(2),\dots,\sigma(k)\}$ (where $P_0(\sigma) \coloneqq \emptyset$) and corresponding \emph{suffix set} as $S_k(\sigma) \coloneqq A \setminus P_k(\sigma)$.

\paragraph{Selection Sort}
Let $a_i$ denote the most preferred alternative according to $\tau$ in the set $S_{k-1}(\sigma^{(k)})$. At step $k$ of selection sort, the alternative $a_i$ is promoted to a position such that the top $k$ alternatives in the new list are ordered according to $\tau$. Note that this step is well-defined only under the \emph{sorted-prefix property}, i.e., at the beginning of step $k$ of the algorithm, the alternatives in $P_{k-1}(\sigma^{(k)})$ are sorted according to $\tau$. This property is maintained by selection sort.

\paragraph{Insertion Sort}
Let $a_i$ denote the most preferred alternative in $S_{k-1}(\sigma^{(k)})$ according to $\sigma^{(k)}$. At step $k$ of insertion sort, the alternative $a_i$ is promoted to a position such that the top $k$ alternatives in the new list are ordered according to $\tau$. Note that this step is well-defined only under the sorted-prefix property, which is maintained by insertion sort.

\paragraph{Sorting Algorithms}
In this work, we will be concerned with sorting algorithms that involve a \emph{combination} of insertion and selection sort. Specifically, we will use the term \emph{sorting algorithm} to refer to a sequence of steps $s_1,s_2,\dots$ such that each step $s_k$ corresponds to either selection or insertion sort, i.e., $s_k \in \{\text{SEL,INS}\}$ for every $k$. If $s_k = \text{SEL}$, then the algorithm promotes the most preferred alternative in $S_{k-1}(\sigma^{(k)})$ (according to $\tau$) to a position such that the top $k$ alternatives in the new list are ordered according to $\tau$. If $s_k = \text{INS}$, then the algorithm promotes the most preferred alternative in $S_{k-1}(\sigma^{(k)})$ (according to $\sigma^{(k)}$) to a position such that the top $k$ alternatives in the new list are ordered according to $\tau$.

For example, in Figure~\ref{fig:ui}, starting from the recommended list at the extreme left, the user performs a \emph{selection sort} operation (promoting 19 to the top of the current list) followed by an \emph{insertion sort} operation (promoting 30 to its correct position in the sorted prefix $\{19,22,40\}$) followed by either selection or insertion sort operation (promoting 23 to its correct position). We will denote a generic sorting algorithm by $\A$ and the class of all sorting algorithms by $\mathfrak{A}$.

\paragraph{Count Function}
Given a sorting algorithm $\A$, a source ranking $\sigma \in \ml(A)$ and a target ranking $\tau \in \ml(A)$, the \emph{count} function $f_{\A}^{\sigma \rightarrow \tau}: [m-1] \rightarrow \Z_+ \cup \{0\}$ keeps track of the \emph{number} of drag-and-drop operations (and the number of \emph{positions} by which some alternative is moved in each such operation) during the execution of $\A$. Formally, $f_{\A}^{\sigma \rightarrow \tau}(\ell)$ is the number of times some alternative is `moved up by $\ell$ positions' during the execution of algorithm $\A$ when the source and target rankings are $\sigma$ and $\tau$ respectively.\footnote{Notice that we do not keep track of \emph{which} alternative is moved by $\ell$ positions. Indeed, we believe it is reasonable to assume that moving the alternative $a_1$ up by $\ell$ positions takes the same time as it will for $a_2$. Also, we do not need to define the count function for \emph{move down} operations as neither selection sort nor insertion sort will ever make such a move.} For example, let $\A$ be insertion sort, $\sigma = (d,c,a,b)$, and $\tau = (a,b,c,d)$. In step 1, the user considers the alternative $d$ and no move-up operation is required. In step 2, the user promotes $c$ by one position (since $c \>_\tau d$) to obtain the new list $(c,d,a,b)$. In step 3, the user promotes $a$ by two positions to obtain $(a,c,d,b)$. Finally, the user promotes $b$ by two positions to obtain the target list $(a,b,c,d)$. Overall, the user performs one `move up by one position' operation and two `move up by two positions' operations. Hence, $f_\A^{\sigma \rightarrow \tau}(1) = 1$, $f_\A^{\sigma \rightarrow \tau}(2) = 2$, and $f_\A^{\sigma \rightarrow \tau}(3) = 0$. We will write $\si$ to denote the total number of drag-and-drop operations performed during the execution of $\A$, i.e., $\si{} = \sum_{\ell=1}^{m-1} f_{\A}^{\sigma \rightarrow \tau}(\ell)$.

\begin{remark}
    Notice the difference between the number of drag-and-drop operations ($\si$) and the total distance covered (i.e., the number of positions by which alternatives are moved). Indeed, the above example involves three drag-and-drop operations ($\si = 3$), but the total distance moved is $0+1+2+2 = 5$. The latter quantity is equal to $\dkt(\sigma,\tau)$.
\label{rem:NumberOfOperations_vs_Distance}
\end{remark}

\begin{restatable}{lemma}{EquivalenceOfSortingAlgorithms}
	For any two sorting algorithms $\A,\A' \in \mathfrak{A}$, any $\sigma,\tau \in \ml(A)$, and any $\ell \in [m-1]$, $f_{\A}^{\sigma \rightarrow \tau}(\ell) = f_{\A'}^{\sigma \rightarrow \tau}(\ell)$.
	\label{lem:EquivalenceOfSortingAlgorithms}
\end{restatable}

In light of \Cref{lem:EquivalenceOfSortingAlgorithms}, we will hereafter drop the subscript $\A$ and simply write $f^{\sigma \rightarrow \tau}$ instead of $f_{\A}^{\sigma \rightarrow \tau}$. The proof of \Cref{lem:EquivalenceOfSortingAlgorithms} appears in \Cref{subsec:Proof_EquivalenceOfSortingAlgorithms}.

\subsection{Drag-and-Drop Time}
\label{subsec:Drag_and_drop}

\paragraph{Weight function}
The \emph{weight} function $\w: [m-1] \rightarrow \R_{\geq 0}$ models the time taken for each drag-and-drop operation; specifically, $\w(\ell)$ denotes the time taken by the user in moving an alternative up by $\ell$ positions.\footnote{Here, `time taken' includes the time spent in \emph{thinking} about which alternative to move as well as actually \emph{carrying out} the move.} Of particular interest to us will be the \emph{linear} weight function $\wlin(\ell) = \ell$ for each $\ell \in [m-1]$ and the \emph{affine} weight function $\waff(\ell) = c\ell + d$ for each $\ell \in [m-1]$ and fixed constants $c,d \in \N$.%, and the \emph{ramp} weight function $\w_{\mathrm{\textup{ramp}}}(\ell) = 1$ for each $\ell \in [m-1]$.

\paragraph{Time Function}
Given the count function $f^{\sigma \rightarrow \tau}$ and the weight function $\w$, the \emph{time} function is defined as their inner product, i.e., 
$\cost_\w(\sigma,\tau) = \langle f^{\sigma \rightarrow \tau}, \w \rangle = \sum_{\ell = 1}^{m-1} f^{\sigma \rightarrow \tau}(\ell) \cdot \w(\ell).$

\Cref{thm:LinearWeightKTdist} shows that for the linear weight function $\wlin$, time is equal to the Kendall's Tau distance, and for the affine weight function, time is equal to a weighted combination of Kendall's Tau distance and the total number of moves.

\begin{restatable}{theorem}{LinearWeightKTdist}
	For any $\sigma,\tau \in \ml(A)$, $\cost_{\wlin}(\sigma,\tau) = \dkt(\sigma,\tau)$ and $\cost_{\waff}(\sigma,\tau) = c \cdot \dkt(\sigma,\tau) + d \cdot \si$.
	\label{thm:LinearWeightKTdist}
\end{restatable}

The proof of \Cref{thm:LinearWeightKTdist} appears in \Cref{subsec:Proof_LinearWeightKTdist}.

\section{Formulation of Recommendation Problem and Theoretical Results}
\label{sec:results}

We model the recommendation problem as the following computational problem: Given the \emph{preference distribution} $\D$ of the user and her $\cost$ function (which, in turn, is determined by the weight function $\w$), find a ranking that minimizes the expected time taken by the user to transform the recommended ranking $\sigma$ into her preference $\tau$.

\begin{definition}[\Rec{}]
	Given a distribution $\D$ over $\ml(A)$, a weight function $\w$, and a number $\delta \in \Q$, does there exist $\sigma \in \ml(A)$ so that $\E_{\tau \sim \D} [\cost_\w(\sigma,\tau)] \leq \delta$? 
\end{definition}

We will focus on settings where the distribution $\D$ is Plackett-Luce, Mallows, or Uniform, and the weight function $\w$ is Linear, Affine, or General. Note that if the quantity $\E_{\tau \sim \D} [\cost_\w(\sigma,\tau)]$ can be computed in polynomial time for a given distribution $\D$ and weight function $\w$, then \Rec{} is in \NP{}.

\begin{table*}
	\footnotesize
	\centering
	\begin{tabular}%{ccc}
		{>{\centering}m{0.17\textwidth} >{\centering}m{0.17\textwidth} >{\centering}m{0.17\textwidth} >{\centering}m{0.21\textwidth} >{\centering\arraybackslash}m{0.13\textwidth}}
		\toprule
		\multirow{2}{*}{\textbf{Distribution $\D$}} & \multicolumn{3}{c}{\textbf{Linear Weights}} & \multicolumn{1}{c}{\textbf{General Weights}}\\
		\cmidrule(l{5pt}r{5pt}){2-4} \cmidrule(l{5pt}r{5pt}){5-5}  
		& Hardness & Exact Algo. & Approx. Algo. & Approx. Algo.\\
		\midrule
		$k$-mixture Plackett-Luce (\kPL) & \NPc{} even for $k=4$ \hspace{1.5cm}  (\Cref{thm:HardnessResults}) & \P{} for $k = 1$ \hspace{1.5cm}  (\Cref{thm:ExactAlgorithms}) & \PTAS{} (\Cref{thm:RecPTAS}) \quad $5$-approx. (\Cref{thm:BordaApproxAlgo}) & $\alpha \beta$-approx. \hspace{1.5cm}  (\Cref{thm:ApproxGeneralWeights}) \\
		\cmidrule{1-5}
		$k$-mixture Mallows (\kMM) & \NPc{} even for $k=4$ \hspace{1.5cm}  (\Cref{thm:HardnessResults}) & \P{} for $k = 1$ \hspace{1.5cm}  (\Cref{thm:ExactAlgorithms}) &  \PTAS{} (\Cref{thm:RecPTAS}) \quad $5$-approx. (\Cref{thm:BordaApproxAlgo}) & $\alpha \beta$-approx. \hspace{1.5cm}  (\Cref{thm:ApproxGeneralWeights}) \\
		\cmidrule{1-5}
		Uniform (\Unif) & \NPc{} even for $n=4$ \hspace{1.5cm}  (\Cref{thm:HardnessResults}) & \P{} for $n \in \{1,2\}$ \hspace{1.5cm}  (\Cref{thm:ExactAlgorithms}) & \PTAS{} (\Cref{thm:RecPTAS}) \quad $5$-approx. (\Cref{thm:BordaApproxAlgo}) & $\alpha \beta$-approx. \hspace{1.5cm}  (\Cref{thm:ApproxGeneralWeights}) \\
		\bottomrule
	\end{tabular}
	\vspace{0.05in}
	\caption{Computational complexity results for \Rec{}. Each row corresponds to a preference model and each column corresponds to a weight function. We use the shorthands \P{}, \NPc{}, \PTAS{}, and $\alpha \beta$-approx. to denote polynomial-time (exact) algorithm, \NPC{}, polynomial-time approximation scheme, and $\alpha \beta$-approximation algorithm respectively. The parameters $\alpha$ and $\beta$ capture how closely a given weight function approximates a \emph{linear} weight function; see \Cref{defn:ClosenessOfWeights}.}
	\label{tab:Results}
\end{table*}

Our computational results for \Rec{} are summarized in \Cref{tab:Results}. We show that this problem is \NPH{}, even when the weight function is linear (\Cref{thm:HardnessResults}). On the algorithmic side, we provide a polynomial-time approximation scheme (PTAS) and a $5$-approximation algorithm for the linear weight function (\Cref{thm:RecPTAS,thm:BordaApproxAlgo}), and an approximation scheme for non-linear weights (\Cref{thm:ApproxGeneralWeights}). 
%All missing proofs can be found in the supplementary material.

\begin{restatable}[\textbf{Exact Algorithms}]{theorem}{ExactAlgorithms}
	\Rec{} is solvable in polynomial time when $\w$ is linear and $\D$ is either (a) $k$-mixture Plackett-Luce $(\kPL)$ with $k = 1$, (b) $k$-mixture Mallows model $(\kMM{})$ with $k = 1$, or (c) a uniform distribution with support size $n \leq 2$.
	\label{thm:ExactAlgorithms}
\end{restatable}

\begin{restatable}[\textbf{Hardness results}]{theorem}{HardnessResults}
	\Rec{} is \NPC{} even when $\w$ is linear and $\D$ is either (a) $k$-mixture Plackett-Luce model $(\kPL{})$ for $k=4$, (b) $k$-mixture Mallows model $(\kMM{})$ for $k=4$, or (c) a uniform distribution over $n=4$ linear orders.
	\label{thm:HardnessResults}
\end{restatable}

\begin{restatable}[\textbf{PTAS}]{theorem}{RecPTAS}
	\Rec{} admits a polynomial time approximation scheme $(\PTAS{})$ when $\w$ is linear and $\D$ is either (a) $k$-mixture Plackett-Luce model $(\kPL{})$ for $k \in \N$, (b) $k$-mixture Mallows model $(\kMM{})$ for $k \in \N$, or (c) a uniform distribution $(\Unif)$.
	\label{thm:RecPTAS}
\end{restatable}

The \PTAS{} in \Cref{thm:RecPTAS} is quite complicated and is primarily of theoretical interest (indeed, for any fixed $\varepsilon > 0$, the running time of the algorithm is ${m}^{2^{\Otilde(1/\varepsilon)}}$, making it difficult to be applied in experiments). A simpler and more practical algorithm (although with a worse approximation) is based on the well-known Borda count voting rule (\Cref{thm:BordaApproxAlgo}).

\begin{restatable}[\textbf{$5$-approximation}]{theorem}{BordaApproxAlgo}
	\Rec{} admits a polynomial time $5$-approximation algorithm when $\w$ is linear and $\D$ is either (a) $k$-mixture Plackett-Luce model $(\kPL{})$ for $k \in \N$, (b) $k$-mixture Mallows model $(\kMM{})$ for $k \in \N$, or (c) a uniform distribution $(\Unif)$.
	\label{thm:BordaApproxAlgo}
\end{restatable}

Our next result (\Cref{thm:ApproxGeneralWeights}) provides an approximation guarantee for \Rec{} that applies to \emph{non-linear} weight functions, as long as they are ``close'' to the linear weight function in the following sense:

\begin{definition}[Closeness-of-weights]
	\label{defn:ClosenessOfWeights}
	A weight function $\w$ is said to be $(\alpha,\beta)$-close to another weight function $\w'$ if there exist $\alpha,\beta \geq 1$ such that for every $\ell \in [m-1]$, we have
	$$\w'(\ell) / \beta \leq \w(\ell) \leq \alpha \w'(\ell).$$
\end{definition}

For \emph{any} (possibly non-linear) weight function $\w$ that is $(\alpha,\beta)$ close to the linear weight function $\wlin$, \Cref{thm:ApproxGeneralWeights} provides an $\alpha \beta$-approximation scheme for \Rec{}.

\begin{restatable}[Approximation for general weights]{theorem}{ApproxGeneralWeights}
	Given any $\varepsilon > 0$ and any weight function $\w$ that is $(\alpha,\beta)$-close to the linear weight function $\wlin$, there exists an algorithm that runs in time $m^{2^{\Otilde(1/\varepsilon)}}$ and returns a linear order $\sigma$ such that
	$$\E_{\tau \sim \D} [\cost_\w(\sigma,\tau)] \leq \alpha \beta (1+\varepsilon) \E_{\tau \sim \D} [\cost_\w(\sigma^*,\tau)],$$
	where $\sigma^* \in \arg\min_{\sigma' \in \ml(A)} \E_{\tau \sim \D} [\cost_\w(\sigma',\tau)]$.
	\label{thm:ApproxGeneralWeights}
\end{restatable}

\begin{remark}
	Notice that the \PTAS{} of \Cref{thm:RecPTAS} is applicable for any \emph{affine} weight function $\waff{} = c \cdot \wlin{} + d$ for some fixed constants $c,d \in \N$. As a result, the approximation guarantee of \Cref{thm:ApproxGeneralWeights} also extends to any weight function that is $(\alpha,\beta)$-close to \emph{some} affine weight function.
\end{remark}

\section{Experimental Results}
\label{sec:exp}
We perform two sets of experiments on Amazon Mechanical Turk (MTurk). The first set of experiments (\Cref{subsec:User_Behavior}) is aimed at identifying the \emph{sorting strategies} of the users as well as a model of their \emph{drag-and-drop behavior}. The observations from these experiments directly motivate the formulation of our theoretical model, which we have already presented in \Cref{sec:results}. The second set of experiments (\Cref{subsec:Expt_Evaluation}) is aimed at \emph{evaluating} the practical usefulness of our approach.

In both sets of experiments, the crowdworkers were asked to sort in increasing order a randomly generated list of numbers between 0 and 100 (the specifics about the length of the lists and how they are generated can be found in \Cref{subsec:User_Behavior,subsec:Expt_Evaluation}). \Cref{fig:mturk} shows an example of the instructions provided to the crowdworkers.

\begin{figure}[ht]
	\centering
	\includegraphics[width=0.6\linewidth]{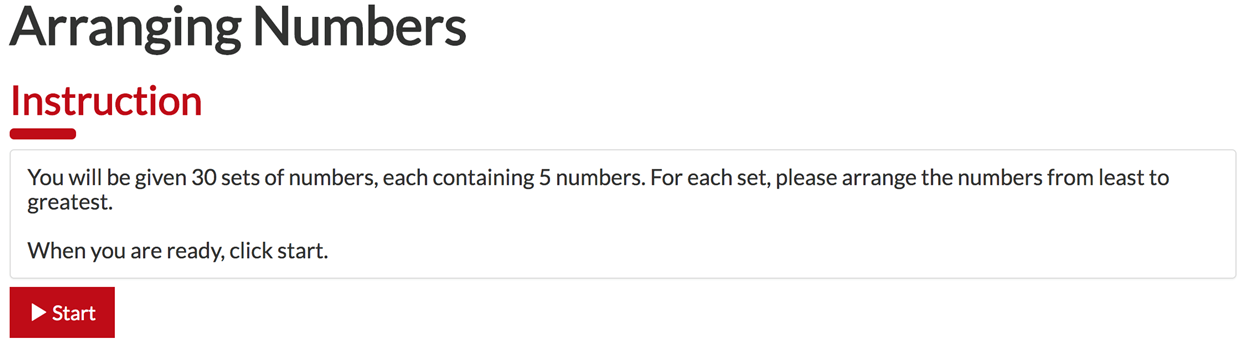}
	\caption{\label{fig:mturk} \small Instructions given to the MTurk workers.} 
\end{figure}

In each experiment, the task length was advertised as 10 minutes, and the payment offered was $\$0.25$ per task. The crowdworkers were provided a user interface (see Figure~\ref{fig:ui}) that allows for drag-and-drop operations. To ensure data quality, we removed those workers from the data who failed to successfully order the integers more than 80\% of the time, or did not complete all the polls. We also removed the workers with high variance in their sorting time; in particular, those with coefficient of variation above the $80^\text{th}$ percentile. The reported results are for the workers whose data was retained.

\subsection{Identifying User Behavior}
\label{subsec:User_Behavior}
To identify user behavior, we performed two experiments: (a) {\sc Rank10}, where each crowdworker participated in 20 polls, each consisting of a list of 10 integers (between $0$ and $100$) generated uniformly at random, and (b) {\sc Rank5}, which is a similar task with 30 polls and lists of length 5. In each poll, we recorded the time taken by a crowdworker to move an alternative (via drag-and-drop operation) and the number of positions by which the alternative was moved. After the initial pruning (as described above), we retained 9840 polls submitted by 492 workers in the {\sc Rank10} experiment, and 10320 polls submitted by 344 workers retained in the {\sc Rank5} experiment. \Cref{tbl:exp} summarizes the aggregate statistics. Our observations are discussed below.

\begin{table*}[ht]%[htp]
	\centering
	\scriptsize
	\begin{tabular}{|c|c|c|c|c|c|c|}
		\hline
		\multirow{2}{*}{}& \multicolumn{3}{c|}{\sc Rank10} & \multicolumn{3}{c|}{\sc Rank5} \\ \cline{2-7}
		& Mean & Median & Std. Dev. & Mean & Median & Std. Dev. \\ \hline
		Sorting time & 24.41 & 22.65 & 9.12 & 7.75 & 6.99 & 3.54 \\ \hline
		Total number of drag-and-drop operations & 7.69 & 8 & 1.8 & 2.91 & 3 & 1.13 \\ \hline
		Total number of positions moved during drag-and-drop operations & 22.59 & 23 & 5.59 & 5.05 & 5 & 2.01 \\ \hline
		Number of operations coinciding with selection/insertion sort & 5.09 & 6 & 2.28 & 2.21 & 2 & 1.06 \\ \hline
		Kendall's Tau distance between the initial and final rankings & 22.55 & 22 & 5.6 & 5.04 & 5 & 2.01 \\ \hline
	\end{tabular}
	\caption{\label{tbl:exp} Summary of the user statistics recorded in the experiments in \Cref{subsec:User_Behavior}.}
\end{table*}

\paragraph{Sorting Behavior}
Our hypothesis regarding the ranking behavior of human crowdworkers was that they use (some combination of) natural sorting algorithms such as selection sort or insertion sort (\Cref{subsec:Sorting_Algorithms}). To test our hypothesis, we examined the fraction of the drag-and-drop operations that \emph{coincided} with an iteration of selection/insertion sort. (Given a ranking $\sigma$, a drag-and-drop operation on $\sigma$ coincides with selection/insertion sort if the order of alternatives resulting from the drag-and-drop operation exactly matches the order of alternatives when one iteration of either selection or insertion sort is applied on $\sigma$.) We found that, on average, $\frac{2.21}{2.91} = 76\%$ of all drag-and-drop operations in {\sc Rank5} (and $\frac{5.09}{7.69} = 66.2\%$ in the {\sc Rank10}) coincided with selection/insertion sort.

\paragraph{Drag-and-Drop Behavior}
To identify the drag-and-drop behavior of the users, we plot the time-to-rank as a function of the total number of positions by which the alternatives are moved in each poll (Figure~\ref{fig:pos_time}). Recall from \Cref{rem:NumberOfOperations_vs_Distance} that for an ideal user who uses only insertion/selection sort, the latter quantity is equal to $\dkt(\sigma,\tau)$.

\begin{figure}[ht]
	\centering
    \includegraphics[scale=0.35]{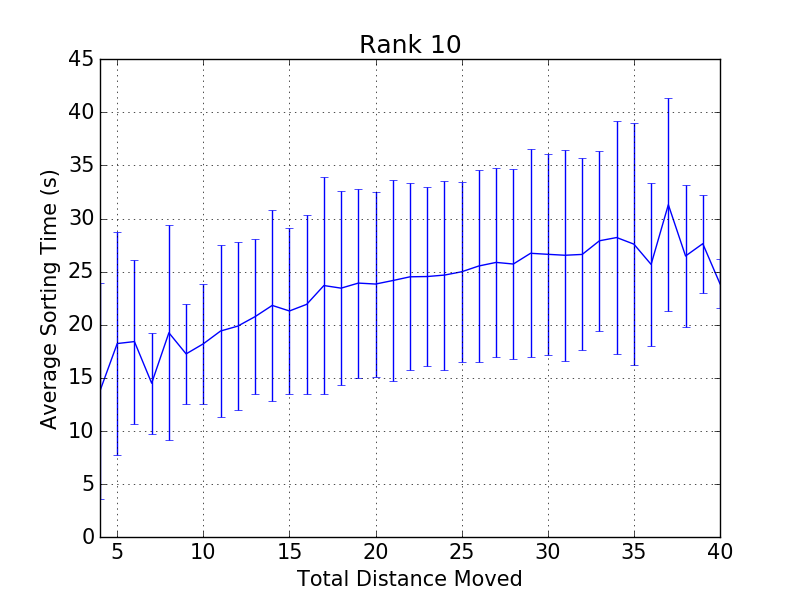}
    \includegraphics[scale=0.35]{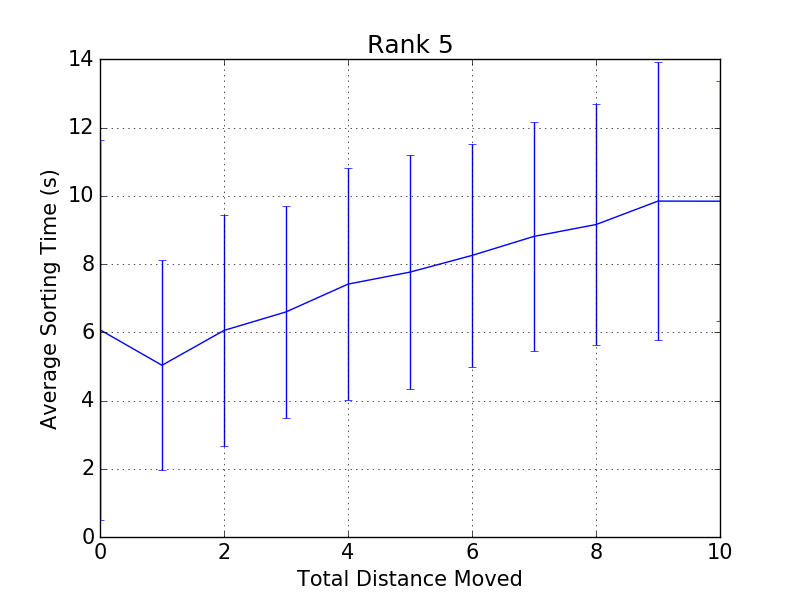}
	\caption{Relationship between the number of positions moved and the total sorting time for {\sc Rank10} (left) and {\sc Rank5} (right).}
	\label{fig:pos_time}
\end{figure}

\begin{table*}[ht]%[htp]
\centering
\footnotesize
\begin{tabular}{|c|c|c|c|c|c|c|} \hline
\multirow{2}{*}{Dataset} & Avg. MSE & $\sqrt{\text{Avg. MSE}}$ & Avg. Sorting Time & \multicolumn{3}{c|}{Number of users based on their best-fit model}  \\
\cline{5-7}
& (in seconds$^2$) & (in seconds) & (in seconds) & Only $\dkt$ & Only $\si$ & Both $\dkt$ and $\si$ \\ \hline
{\sc Rank10} & 42.98 & 6.56 & 24.41 & 217 & 199 & 76 \\ \hline
{\sc Rank5} & 7.74 & 2.78 & 7.75 & 138 & 180 & 26 \\ \hline
\end{tabular}
\caption{\label{tbl:mse}\footnotesize Average 5-fold cross-validation MSE over all workers using the best model for each worker, and the number of users for which each of the models was identified to be the best. \# moves is the number of times alternatives are moved using selection or insertion sort.}
\end{table*}

Our hypothesis was that the sorting time varies \emph{linearly} with the total number of drag-and-drop operations ($\si$) and the Kendall's Tau distance ($\dkt(\sigma,\tau)$). To verify this, we used linear regression with time-to-rank (or sorting time) as the target variable and measured the mean squared error (MSE) using 5-fold cross-validation for three different choices of independent variables: (1) Only $\dkt{}$, (2) only $\si$, and (3) both $\dkt{}$ and $\si$. For each user, we picked the model with the smallest MSE (see \Cref{tbl:mse} for the resulting distribution of the number of users). We found that the predicted drag-and-drop times (using the best-fit model for each user) are, on average, within $\frac{6.56}{24.41} = 26.8\%$ of the observed times for {\sc Rank10} and within $\frac{2.78}{7.75} = 35.8\%$ for {\sc Rank5}.

\subsection{Evaluating the Proposed Framework}
\label{subsec:Expt_Evaluation}
To evaluate the usefulness of our framework, we compared a random recommendation strategy with one that forms an increasingly accurate estimate of users' preferences with time. Specifically, we first fix the ground truth ranking of $10$ alternatives consisting of randomly generated integers between $0$ and $100$. Each crowdworker then participates in two sets of $10$ polls each. In one set of polls, the crowdworkers are provided with initial rankings generated by adding independent Gaussian noise to the ground truth (to simulate a \emph{random} recommendation strategy), and their sorting times are recorded.

In the second set of polls, the recommended set of alternatives is the same as under the random strategy but ordered order to a \emph{Borda} ranking. Specifically, the ordering in the $k^\text{th}$ iteration is determined by the Borda ranking aggregated from the previous $k-1$ iterations.
%we simulate a \emph{Borda} recommendation strategy as follows: The alternatives were shown in random order in the first poll, and in each subsequent $k$-th poll, the alternatives were sorted using the Borda voting rule to aggregate $k-1$ previous votes.

\begin{figure}[ht]
	\centering
    \includegraphics[scale=0.35]{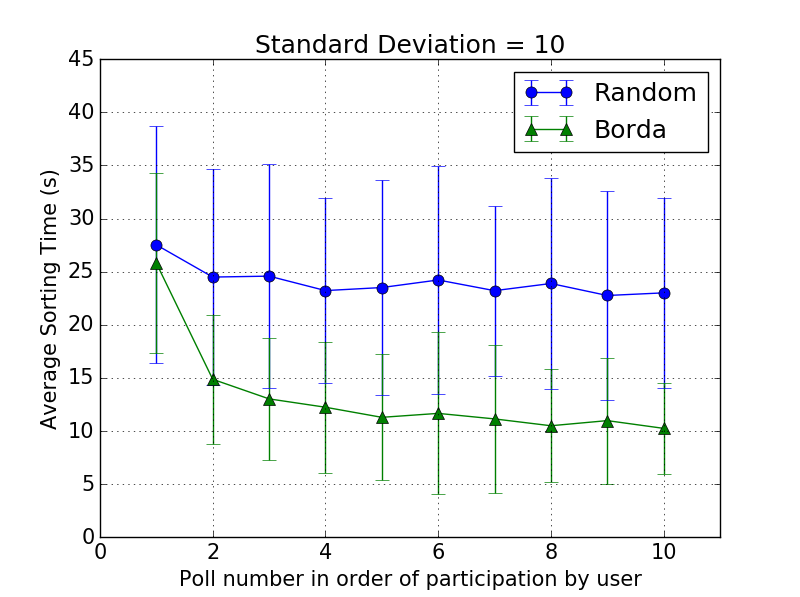}
    \includegraphics[scale=0.35]{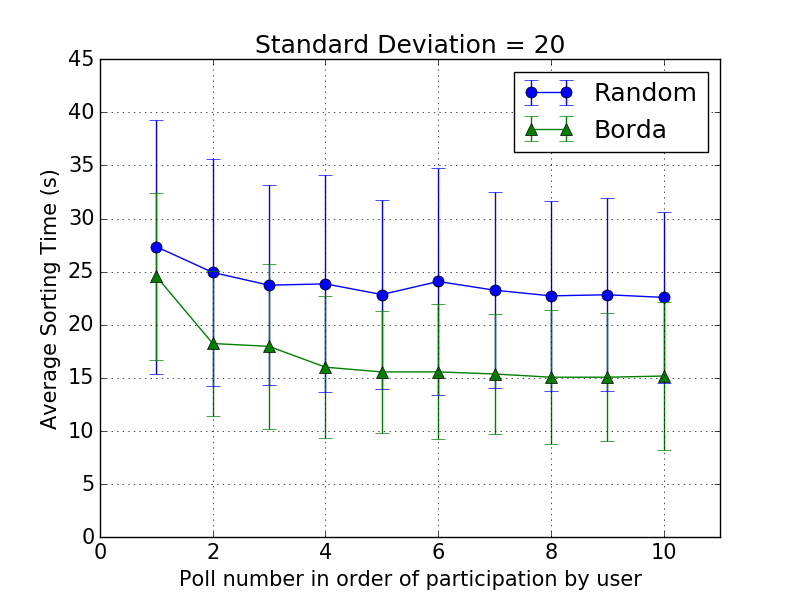}
	\caption{Relationship between sorting time and the number of polls completed by the users for std dev=$10$ (left) and std dev=$20$ (right).}
	\label{fig:reco}
\end{figure}

\Cref{fig:reco} shows the average sorting time of the crowdworkers as a function of the index of the polls under two different noise settings: std. dev. = 10 and std. dev. = 20. We can make two important observations: First, that Borda recommendation strategy (in green) provides a significant reduction in the sorting time of the users compared to the random strategy (in blue). Indeed, the sorting time of the users is reduced by \emph{up to $50\%$}, thus validating the practical usefulness of our framework. The second observation is that the reduction in sorting time is \emph{not} due to increasing familiarity with the interface. This is because the average sorting time for the random strategy remains almost constant throughout the duration of the poll.

\section{Conclusion and Future Work}
We proposed a recommendation framework to minimize time-to-rank. We formulated a theoretical model of the recommendation problem (including \NPH{}ness results and associated approximation algorithms), and illustrated the practical effectiveness of our approach in real-world experiments. 

Our work opens up a number of directions for future research. In terms of theoretical questions, it would be interesting to analyze the complexity of the recommendation problem for other distance measures, e.g., Ulam distance. On the practical side, it would be interesting to analyze the effect of cognitive biases such as the \emph{framing effect} \citep{TK81framing} and \emph{list position bias} \citep{LH14leveraging} on the recommendation problem. %\citep{ABH17controlling}
 Progress in this direction can, in turn, have implications on the \emph{fairness} of recommendation algorithms.

\section*{Acknowledgments}
We are grateful to IJCAI-19 reviewers for their helpful comments. This work is supported by NSF \#1453542 and ONR \#N00014-17-1-2621.

\bibliographystyle{plainnat}
\bibliography{references}

\begin{thebibliography}{36}
\providecommand{\natexlab}[1]{#1}
\providecommand{\url}[1]{\texttt{#1}}
\expandafter\ifx\csname urlstyle\endcsname\relax
  \providecommand{\doi}[1]{doi: #1}\else
  \providecommand{\doi}{doi: \begingroup \urlstyle{rm}\Url}\fi

\bibitem[Agarwal(2016)]{Agarwal2016:On-ranking}
Shivani Agarwal.
\newblock {On Ranking and Choice Models}.
\newblock In \emph{Proceedings of the Twenty-Fifth International Joint
  Conference on Artificial Intelligence}, pages 4050--4053, 2016.

\bibitem[Ailon et~al.(2008)Ailon, Charikar, and Newman]{Ailon08:Aggregating}
Nir Ailon, Moses Charikar, and Alantha Newman.
\newblock {Aggregating Inconsistent Information: Ranking and Clustering}.
\newblock \emph{Journal of the ACM}, 55\penalty0 (5):\penalty0 23, 2008.

\bibitem[Albert and Tullis(2013)]{Albert13:Measuring}
William Albert and Thomas Tullis.
\newblock \emph{{Measuring the User Experience: Collecting, Analyzing, and
  Presenting Usability Metrics}}.
\newblock Newnes, 2013.

\bibitem[Alon(2006)]{Alon06:Ranking}
Noga Alon.
\newblock {Ranking Tournaments}.
\newblock \emph{SIAM Journal on Discrete Mathematics}, 20\penalty0
  (1):\penalty0 137--142, 2006.

\bibitem[Awasthi et~al.(2014)Awasthi, Blum, Sheffet, and
  Vijayaraghavan]{Awasthi14:Learning}
Pranjal Awasthi, Avrim Blum, Or~Sheffet, and Aravindan Vijayaraghavan.
\newblock {Learning Mixtures of Ranking Models}.
\newblock In \emph{Proceedings of Advances in Neural Information Processing
  Systems}, pages 2609--2617, 2014.

\bibitem[Bevan et~al.(2015)Bevan, Carter, and Harker]{Bevan15:ISO}
Nigel Bevan, James Carter, and Susan Harker.
\newblock {ISO 9241-11 revised: What have we Learnt about Usability Since
  1998?}
\newblock In \emph{International Conference on Human-Computer Interaction},
  pages 143--151. Springer, 2015.

\bibitem[Blum et~al.(2004)Blum, Jackson, Sandholm, and
  Zinkevich]{Blum04:Preference}
Avrim Blum, Jeffrey Jackson, Tuomas Sandholm, and Martin Zinkevich.
\newblock {Preference Elicitation and Query Learning}.
\newblock \emph{Journal of Machine Learning Research}, 5:\penalty0 649--667,
  2004.

\bibitem[Boutilier(2013)]{Boutilier13:Computational}
Craig Boutilier.
\newblock {Computational Decision Support: Regret-Based Models for Optimization
  and Preference Elicitation}.
\newblock In P.~H. Crowley and T.~R. Zentall, editors, \emph{{Comparative
  Decision Making: Analysis and Support Across Disciplines and Applications}}.
  Oxford University Press, 2013.

\bibitem[Braverman and Mossel(2008)]{Braverman08:Noisy}
Mark Braverman and Elchanan Mossel.
\newblock {Noisy Sorting Without Resampling}.
\newblock In \emph{Proceedings of the Nineteenth Annual ACM-SIAM Symposium on
  Discrete Algorithms}, pages 268--276, 2008.

\bibitem[Busa-Fekete et~al.(2014)Busa-Fekete, H{\"u}llermeier, and
  Sz{\"o}r{\'e}nyi]{Busa-Fekete2014:Preference-based}
R{\'o}bert Busa-Fekete, Eyke H{\"u}llermeier, and Bal{\'a}zs Sz{\"o}r{\'e}nyi.
\newblock {Preference-Based Rank Elicitation using Statistical Models: The Case
  of Mallows}.
\newblock In \emph{Proceedings of the 31st International Conference on
  International Conference on Machine Learning}, pages II:1071--1079, 2014.

\bibitem[Charon and Hudry(2010)]{Charon10:Updated}
Ir\`{e}ne Charon and Olivier Hudry.
\newblock {An Updated Survey on the Linear Ordering Problem for Weighted or
  Unweighted Tournaments}.
\newblock \emph{Annals of Operations Research}, 175\penalty0 (1):\penalty0
  107--158, 2010.

\bibitem[Cheng et~al.(2015)Cheng, Teevan, and Bernstein]{Cheng15:Measuring}
Justin Cheng, Jaime Teevan, and Michael~S Bernstein.
\newblock {Measuring Crowdsourcing Effort with Error-Time Curves}.
\newblock In \emph{Proceedings of the 33rd Annual ACM Conference on Human
  Factors in Computing Systems}, pages 1365--1374, 2015.

\bibitem[Conen and Sandholm(2001)]{Conen01:Minimal}
Wolfram Conen and Tuomas Sandholm.
\newblock {Minimal Preference Elicitation in Combinatorial Auctions}.
\newblock In \emph{IJCAI-2001 Workshop on Economic Agents, Models, and
  Mechanisms}, pages 71--80, 2001.

\bibitem[Conitzer(2006)]{Conitzer06:Slater}
Vincent Conitzer.
\newblock {Computing Slater Rankings Using Similarities among Candidates}.
\newblock In \emph{Proceedings of the 21st National Conference on Artificial
  Intelligence}, volume~1, pages 613--619, 2006.

\bibitem[Conitzer and Sandholm(2002)]{Conitzer02:Elicitation}
Vincent Conitzer and Tuomas Sandholm.
\newblock {Vote Elicitation: Complexity and Strategy-Proofness}.
\newblock In \emph{Eighteenth National Conference on Artificial Intelligence},
  pages 392--397, 2002.

\bibitem[Coppersmith et~al.(2010)Coppersmith, Fleischer, and
  Rurda]{CFR10ordering}
Don Coppersmith, Lisa~K Fleischer, and Atri Rurda.
\newblock {Ordering by Weighted Number of Wins Gives a Good Ranking for
  Weighted Tournaments}.
\newblock \emph{ACM Transactions on Algorithms (TALG)}, 6\penalty0
  (3):\penalty0 55, 2010.

\bibitem[Dwork et~al.(2001)Dwork, Kumar, Naor, and Sivakumar]{Dwork01:Rank}
Cynthia Dwork, Ravi Kumar, Moni Naor, and D.~Sivakumar.
\newblock {Rank Aggregation Methods for the Web}.
\newblock In \emph{Proceedings of the 10th World Wide Web Conference}, pages
  613--622, 2001.

\bibitem[Guiver and Snelson(2009)]{Guiver09:Bayesian}
John Guiver and Edward Snelson.
\newblock {Bayesian Inference for Plackett-Luce Ranking Models}.
\newblock In \emph{Proceedings of the 26th Annual International Conference on
  Machine Learning}, ICML-09, pages 377--384, Montreal, Quebec, Canada, 2009.

\bibitem[Kendall(1938)]{K38new}
Maurice~G Kendall.
\newblock {A New Measure of Rank Correlation}.
\newblock \emph{Biometrika}, 30\penalty0 (1/2):\penalty0 81--93, 1938.

\bibitem[Kenyon-Mathieu and Schudy(2007)]{Kenyon07:How}
Claire Kenyon-Mathieu and Warren Schudy.
\newblock {How to Rank with Few Errors: A PTAS for Weighted Feedback Arc Set on
  Tournaments}.
\newblock In \emph{Proceedings of the Thirty-Ninth Annual ACM Symposium on
  Theory of Computing}, pages 95--103, 2007.

\bibitem[Khetan and Oh(2016)]{Khetan2016:Data-driven}
Ashish Khetan and Sewoong Oh.
\newblock {Data-Driven Rank Breaking for Efficient Rank Aggregation}.
\newblock \emph{Journal of Machine Learning Research}, 17\penalty0
  (193):\penalty0 1--54, 2016.

\bibitem[Laskowski et~al.(2004)Laskowski, Autry, Cugini, and
  Killam]{Laskowski04:Improving}
Sharon~J Laskowski, Marguerite Autry, John Cugini, and William Killam.
\newblock {Improving the Usability and Accessibility of Voting Systems and
  Products}.
\newblock \emph{NIST Special Publication}, 500:\penalty0 256, 2004.

\bibitem[Lerman and Hogg(2014)]{LH14leveraging}
Kristina Lerman and Tad Hogg.
\newblock {Leveraging Position Bias to Improve Peer Recommendation}.
\newblock \emph{PloS one}, 9\penalty0 (6):\penalty0 e98914, 2014.

\bibitem[Liu(2011)]{Liu11:Learning}
Tie-Yan Liu.
\newblock \emph{Learning to Rank for Information Retrieval}.
\newblock Springer, 2011.

\bibitem[Lu and Boutilier(2014)]{Lu2014:Effective}
Tyler Lu and Craig Boutilier.
\newblock {Effective Sampling and Learning for Mallows Models with
  Pairwise-Preference Data}.
\newblock \emph{Journal of Machine Learning Research}, 15:\penalty0 3963--4009,
  2014.

\bibitem[Luce(1959)]{Luce59:Individual}
Robert~Duncan Luce.
\newblock \emph{{Individual Choice Behavior: A Theoretical Analysis}}.
\newblock Wiley, 1959.

\bibitem[Maddalena et~al.(2016)Maddalena, Basaldella, De~Nart, Degl'Innocenti,
  Mizzaro, and Demartini]{Maddalena16:Crowdsourcing}
Eddy Maddalena, Marco Basaldella, Dario De~Nart, Dante Degl'Innocenti, Stefano
  Mizzaro, and Gianluca Demartini.
\newblock {Crowdsourcing Relevance Assessments: The Unexpected Benefits of
  Limiting the Time to Judge}.
\newblock In \emph{Fourth AAAI Conference on Human Computation and
  Crowdsourcing}, 2016.

\bibitem[Mallows(1957)]{Mallows57:Non-null}
Colin~L. Mallows.
\newblock {Non-Null Ranking Model}.
\newblock \emph{Biometrika}, 44\penalty0 (1/2):\penalty0 114--130, 1957.

\bibitem[Negahban et~al.(2017)Negahban, Oh, and Shah]{Negahban17:Rank}
Sahand Negahban, Sewoong Oh, and Devavrat Shah.
\newblock {Rank Centrality: Ranking from Pairwise Comparisons}.
\newblock \emph{Operations Research}, 65\penalty0 (1):\penalty0 266--287, 2017.

\bibitem[Plackett(1975)]{Plackett75:Analysis}
Robin~L. Plackett.
\newblock {The Analysis of Permutations}.
\newblock \emph{Journal of the Royal Statistical Society. Series C (Applied
  Statistics)}, 24\penalty0 (2):\penalty0 193--202, 1975.

\bibitem[Procaccia and Shah(2016)]{Procaccia16:Optimal}
Ariel~D Procaccia and Nisarg Shah.
\newblock {Optimal Aggregation of Uncertain Preferences}.
\newblock In \emph{Thirtieth AAAI Conference on Artificial Intelligence}, pages
  608--614, 2016.

\bibitem[Soufiani et~al.(2013)Soufiani, Parkes, and Xia]{Azari13:Preference}
Hossein~Azari Soufiani, David~C Parkes, and Lirong Xia.
\newblock {Preference Elicitation For General Random Utility Models}.
\newblock In \emph{Proceedings of the Twenty-Ninth Conference on Uncertainty in
  Artificial Intelligence}, pages 596--605, 2013.

\bibitem[Tversky and Kahneman(1981)]{TK81framing}
Amos Tversky and Daniel Kahneman.
\newblock {The Framing of Decisions and the Psychology of Choice}.
\newblock \emph{Science}, 211\penalty0 (4481):\penalty0 453--458, 1981.

\bibitem[Zhao and Xia(2018)]{Zhao2018:Composite}
Zhibing Zhao and Lirong Xia.
\newblock {Composite Marginal Likelihood Methods for Random Utility Models}.
\newblock In \emph{Proceedings of the 35th International Conference on Machine
  Learning}, volume~80 of \emph{Proceedings of Machine Learning Research},
  pages 5922--5931, 2018.

\bibitem[Zhao et~al.(2016)Zhao, Piech, and Xia]{Zhao16:Learning}
Zhibing Zhao, Peter Piech, and Lirong Xia.
\newblock {Learning Mixtures of Plackett-Luce Models}.
\newblock In \emph{Proceedings of the 33rd International Conference on Machine
  Learning (ICML-16)}, 2016.

\bibitem[Zhao et~al.(2018)Zhao, Li, Wang, Kephart, Mattei, Su, and
  Xia]{Zhao2018:A-Cost-Effective}
Zhibing Zhao, Haoming Li, Junming Wang, Jeffrey Kephart, Nicholas Mattei, Hui
  Su, and Lirong Xia.
\newblock {A Cost-Effective Framework for Preference Elicitation and
  Aggregation}.
\newblock In \emph{Proceedings of the 2018 Conference on Uncertainty in
  Artificial Intelligence (UAI)}, 2018.

\end{thebibliography}

\clearpage
\newpage
\section{Appendix}

\subsection{Additional Preliminaries}
\label{subsec:Appendix_Preliminaries}

The pairwise marginal distribution for the $k$-mixture Plackett-Luce model is given by
\begin{align}
 \Pr_{\sigma \sim \kPL} (a_i \>_\sigma a_j) = \sum_{\ell=1}^k \gamma^{(\ell)} \cdot \frac{\theta^{(\ell)}_i}{\theta^{(\ell)}_i + \theta^{(\ell)}_j}.
\label{eqn:PairwiseMarginalkPL}
\end{align}

\begin{prop}[\citealp{Mallows57:Non-null}]
	Let $\sigma^*,\phi$ be the parameters of a Mallows model (\MM{}), and let $a_i,a_j \in A$ be such that $a_i \>_{\sigma^*} a_j$. Let $\Delta = \rank(\sigma^*,a_j) - \rank(\sigma^*,a_i)$. Then,
	\begin{align*}
	\Pr_{\sigma \sim (\sigma^*,\phi)} (a_i \>_\sigma a_j) = \frac{\sum_{z=1}^{\Delta} z \phi^{z-1}}{\left( \sum_{z=0}^{\Delta-1} \phi^{z} \right) \left( \sum_{z=0}^{\Delta} \phi^{z} \right)}.
	\end{align*}
	\label{prop:MM-pairwise}
\end{prop}

The pairwise marginal for a $k$-mixture Mallows model, parameterized by $\{\gamma^{(\ell)},\sigma^*_{(\ell)},\phi_{(\ell)}\}_{\ell=1}^k$, can be derived similarly. Fix a pair $a_i,a_j \in A$. For each $\ell \in [k]$, let $\Delta^{i,j}_\ell \coloneqq \rank(\sigma^*_{(\ell)},a_j) - \rank(\sigma^*_{(\ell)},a_i)$. Define the function $g_{\ell}: \Z \setminus \{0\} \rightarrow \R_{\geq 0}$ as
\[ g_{\ell}(\Delta) \coloneqq
\begin{cases}
\frac{\sum_{z=1}^{\Delta} z \phi_{(\ell)}^{z-1}}{\left( \sum_{z=0}^{\Delta-1} \phi_{(\ell)}^{z} \right) \left( \sum_{z=0}^{\Delta} \phi_{(\ell)}^{z} \right)}       & \quad \text{if } \Delta > 0,\\
1-\frac{\sum_{z=1}^{|\Delta|} z \phi_{(\ell)}^{z-1}}{\left( \sum_{z=0}^{|\Delta|-1} \phi_{(\ell)}^{z} \right) \left( \sum_{z=0}^{|\Delta|} \phi_{(\ell)}^{z} \right)}  & \quad \text{if } \Delta < 0.
\end{cases}
\]
Thus, $g_{\ell}(\Delta^{i,j}_{\ell})$ is the pairwise marginal probability induced by the $\ell^{\text{th}}$ mixture, i.e.,
\begin{align*}
    g_{\ell}(\Delta^{i,j}_{\ell}) = \Pr_{\sigma \sim ( \sigma^*_{(\ell)},\phi_{(\ell)} )} (a_i \>_{\sigma} a_j).
\end{align*}

The pairwise marginal for the $\kMM$ model is given by
\begin{align}
\Pr_{\sigma \sim \kMM} (a_i \>_\sigma a_j) = \sum_{\ell=1}^k \gamma^{(\ell)} g_{\ell}(\Delta^{i,j}_{\ell}).
\label{eqn:PairwiseMarginalkMM}
\end{align}

\subsection{Relevant Computational Problems}
\label{subsec:Computational_Problems}

\begin{definition}[\Kem]
	Given a preference profile $\{\sigma^{(i)}\}_{i=1}^n$ and a number $\delta \in \Q$, does there exist $\sigma \in \ml(A)$ such that $\sum_{i=1}^n \dkt(\sigma,\sigma^{(i)}) \leq \delta$?
\end{definition}

\Kem{} is known to be \NPC{} even for $n = 4$ \citep{Dwork01:Rank}.

\begin{definition}[\WFAST{} (\wfast{})]
	Given a complete directed graph $G = (V,E)$, a set of non-negative edge weights $\{w_{i,j},w_{j,i}\}_{(i,j) \in E}$ where $w_{i,j} + w_{j,i} = b$ for some fixed constant $b \in (0,1]$, and a number $\delta \in \Q$, does there exist $\sigma \in \ml(A)$ such that $\sum_{i,j \in V} w_{j,i} \cdot \Ind[a_i \>_{\sigma} a_j] \leq \delta$?
\end{definition}

\wfast{} is known to be \NPC{} even when $w_{i,j} = 1$ if $(i,j) \in E$ and $0$ otherwise \citep{Ailon08:Aggregating,Alon06:Ranking,Conitzer06:Slater,Charon10:Updated}. A polynomial-time approximation scheme (\PTAS{}) for \wfast{} is also known \citep{Kenyon07:How}. \Cref{prop:KendallTauPTAS} recalls this result.

\begin{prop}[\citealp{Kenyon07:How}]
	Given any $\varepsilon > 0$ and an instance of \wfast{}, there exists an algorithm that runs in time ${|V|}^{2^{\Otilde(1/\varepsilon)}}$ and returns a linear order $\sigma$ such that
	$$\sum_{i,j \in V} w_{j,i} \cdot \Ind[a_i \>_{\sigma} a_j] \leq (1+\varepsilon) \sum_{i,j \in V} w_{j,i} \cdot \Ind[a_i \>_{\sigma^*} a_j],$$
	where $\sigma^* \in \arg\min_{\tau \in \ml(A)} \sum_{i,j \in V} w_{j,i} \cdot \Ind[a_i \>_{\tau} a_j]$.
	\label{prop:KendallTauPTAS}
\end{prop}

When $b=1$, \wfast{} admits a $5$-approximation algorithm based on the Borda count voting rule (i.e., ordering the vertices in increasing order of their weighted indegrees).

\begin{prop}[\citealp{CFR10ordering}]
	There is a polynomial-time algorithm that, given any instance of \wfast{} with $b=1$, returns a linear order $\sigma$ such that
	\begin{align*}
	\sum_{i,j \in V} w_{j,i} \cdot \Ind[a_i \>_{\sigma} a_j] \leq 5 \cdot \sum_{i,j \in V} w_{j,i} \cdot \Ind[a_i \>_{\sigma^*} a_j],
	\end{align*}
	where $\sigma^* \in \arg\min_{\tau \in \ml(A)} \sum_{i,j \in V} w_{j,i} \cdot \Ind[a_i \>_{\tau} a_j]$.
	\label{prop:BordaApprox}
\end{prop}

\subsection{Proof of Lemma~\ref{lem:EquivalenceOfSortingAlgorithms}}
\label{subsec:Proof_EquivalenceOfSortingAlgorithms}

\EquivalenceOfSortingAlgorithms*
\begin{proof}
We will prove \Cref{lem:EquivalenceOfSortingAlgorithms} via induction on the number of alternatives $m$. The base case of $m=1$ is trivial. Suppose the lemma holds for all alternative sets of size $m \leq n-1$. We will show that the lemma also holds for $m=n$.

Let $\sigma, \tau$ be any two linear orders over the same set of $n$ alternatives, namely $A$. Let $a \coloneqq \tau(1)$ be the most preferred alternative under $\tau$, and let $a$ be ranked $k^\text{th}$ under $\sigma$, i.e., $\sigma(k) = a$. Let $\sigma_{-a}$ and $\tau_{-a}$ denote the truncated linear orders obtained by dropping the alternative $a$ from $\sigma$ and $\tau$ respectively.

We will show that for any sorting algorithm $\A \in \mathfrak{A}$, the following conditions hold:

$\text{ If } k = n, \text{ then } $
\begin{equation}
f_{\A}^{\sigma \rightarrow \tau}(\ell) =
  \begin{cases}
    f_{\A}^{\sigma_{-a} \rightarrow \tau_{-a}}(\ell) & \text{ for all } \ell \in [n-2], \text{ and }\\
    1 & \text{ for } \ell = n-1;
  \end{cases}
\label{eqn:Claim1}
\end{equation}

$\text{ and if } k < n, \text{ then } $
\begin{equation}
f_{\A}^{\sigma \rightarrow \tau}(\ell) =
  \begin{cases}
    f_{\A}^{\sigma_{-a} \rightarrow \tau_{-a}}(\ell) & \text{ for all } \ell \in [n-2] \setminus \{k-1\},\\
    f_{\A}^{\sigma_{-a} \rightarrow \tau_{-a}}(\ell) + 1 & \text{ for } \ell = k-1, \text{ and }\\
    0 & \text{ for } \ell = n-1.
  \end{cases}
\label{eqn:Claim2}
\end{equation}

Note that the claims in \Cref{eqn:Claim1,eqn:Claim2} suffice to prove the lemma: Indeed, $\sigma_{-a}$ and $\tau_{-a}$ are valid linear orders over the same set of $(n-1)$ alternatives, namely $A \setminus \{a\}$. Therefore, by the induction hypothesis, we have that for any two sorting algorithms $\A,\A' \in \mathfrak{A}$ and any $\ell \in [n-2]$,
\begin{equation}
f_{\A}^{\sigma_{-a} \rightarrow \tau_{-a}}(\ell) = f_{\A'}^{\sigma_{-a} \rightarrow \tau_{-a}}(\ell).
\label{eqn:Temp1}
\end{equation}
\Cref{eqn:Claim1,eqn:Claim2,eqn:Temp1} together give us that $f_{\A}^{\sigma \rightarrow \tau}(\ell) = f_{\A'}^{\sigma \rightarrow \tau}(\ell)$ for all $\ell \in [n-1]$, as desired.

To prove the claims in \Cref{eqn:Claim1,eqn:Claim2}, recall from \Cref{subsec:Sorting_Algorithms} that a \emph{sorting algorithm} $\A$ is a sequence of steps $s_1, s_2, \dots$ such that every step corresponds to either selection or insertion sort, i.e., $s_j = \{\text{SEL}, \text{INS}\}$ for every $j$. We will prove the claims via case analysis based on whether $\A$ performs a selection sort operation during the first $k$ steps or not.\\

\textbf{Case I}: \emph{At least one of the first $k$ steps $s_1,\dots,s_k$ is selection sort}.\\

Let $1 \leq i \leq k$ be such that $s_i = \text{SEL}$ and $s_j = \text{INS}$ for all $1 \leq j < i$. In the first $(i-1)$ steps (which are all insertion sort operations), the algorithm $\A$ only considers the top $(i-1)$ alternatives in $\sigma$, namely $P_{i-1}(\sigma)$. Furthermore, since $i-1 < k$, we have that $a \notin P_{i-1}(\sigma)$. Therefore, the top $(i-1)$ alternatives in $\sigma$ are identical to those in $\sigma_{-a}$, and the execution of $\A$ during $\sigma \rightarrow \tau$ is identical to that during $\sigma_{-a} \rightarrow \tau_{-a}$ for the first $(i-1)$ steps. Stated differently, if $f_{\A}^{\sigma,i}(\ell)$ and $ f_{\A}^{\sigma_{-a},i}(\ell)$ denote the number of move-up-by-$\ell$-positions operations performed by $\A$ during the first $i$ steps for the input $\sigma$ and $\sigma_{-a}$ respectively, then 
%for any $\ell \in [n-2]$, the number of move-up-by-$\ell$-positions operations performed by $\A$ for the input $\sigma$ at the end of the the first $(i-1)$ steps equals that for $\sigma_{-a}$.
%at the end of the the first $(i-1)$ steps,
$f_{\A}^{\sigma,i-1}(\ell) = f_{\A}^{\sigma_{-a},i-1}(\ell)$ for all $\ell \in [n-2]$ and $f_{\A}^{\sigma,i-1}(n-1) = 0$.

At the $i^\text{th}$ step, $\A$ performs a selection sort operation. This involves promoting the alternative $a$ by $(k-1)$ positions to the top of the current list. Therefore, at the end of the first $i$ steps, we have:

$\text{ If } k = n, \text{ then } $
\begin{equation}
f_{\A}^{\sigma,i}(\ell) =
  \begin{cases}
    f_{\A}^{\sigma_{-a},i-1}(\ell) & \text{ for all } \ell \in [n-2], \text{ and }\\
    1 & \text{ for } \ell = n-1;
  \end{cases}
\label{eqn:Temp2}
\end{equation}

$\text{ and if } k < n, \text{ then } $
\begin{equation}
f_{\A}^{\sigma,i}(\ell) =
  \begin{cases}
    f_{\A}^{\sigma_{-a},i-1}(\ell) & \text{ for all } \ell \in [n-2] \setminus \{k-1\},\\
    f_{\A}^{\sigma_{-a},i-1}(\ell) + 1 & \text{ for } \ell = k-1, \text{ and }\\
    0 & \text{ for } \ell = n-1.
  \end{cases}
\label{eqn:Temp3}
\end{equation}

%Notice that in \Cref{eqn:Temp2,eqn:Temp3}, the quantity $f_{\A}^{\sigma \rightarrow \tau}(\ell)$ is with respect to the first $i$ steps of the execution of $\A$ during $\sigma \rightarrow \tau$, whereas the quantity $f_{\A}^{\sigma_{-a} \rightarrow \tau_{-a}}(\ell)$ is with respect to the first $(i-1)$ steps of the execution of $\A$ during $\sigma_{-a} \rightarrow \tau_{-a}$.

Let $\sigma'$ denote the list maintained by $\A$ at the end of the $i^\text{th}$ step during $\sigma \rightarrow \tau$. In addition, let $\sigma''$ denote the list maintained by $\A$ at the end of the $(i-1)^\text{th}$ step during $\sigma_{-a} \rightarrow \tau_{-a}$. We therefore have that 
\begin{align}
    f_{\A}^{\sigma \rightarrow \tau}(\ell) = f_{\A}^{\sigma,i}(\ell) + f_{\A}^{\sigma' \rightarrow \tau}(\ell) \text{ for every } \ell \in [n-1],\text{ and } \nonumber \\ f_{\A}^{\sigma_{-a} \rightarrow \tau_{-a}}(\ell) = f_{\A}^{\sigma_{-a},i-1}(\ell) + f_{\A}^{\sigma'' \rightarrow \tau_{-a}}(\ell) \text{ for every } \ell \in [n-2].
\label{eqn:Temp34}
\end{align}

Observe that $\sigma' = (a,\sigma'')$. Consider the execution of $\A$ during $\sigma' \rightarrow \tau$ and during $\sigma'' \rightarrow \tau_{-a}$. From \Cref{lem:SameTopAlternative} (stated below), we have that 
\begin{align}
f_{\A}^{\sigma' \rightarrow \tau}(\ell) = f_{\A}^{\sigma'' \rightarrow \tau_{-a}}(\ell) \text{ for all } \ell \in [n-2] \text{ and } f_{\A}^{\sigma' \rightarrow \tau}(n-1) = 0.
\label{eqn:Temp4}
\end{align}

\Cref{eqn:Temp2,eqn:Temp3,eqn:Temp34,eqn:Temp4} together give the desired claim.\\

\textbf{Case II}: \emph{Each of the first $k$ steps is insertion sort, i.e., $s_1 = \textup{INS}, \dots, s_k = \textup{INS}$}.\\

The analysis in this case is identical to that of Case I for the first $(k-1)$ steps. That is, at the end of the first $(k-1)$ steps, $f_{\A}^{\sigma \rightarrow \tau}(\ell) = f_{\A}^{\sigma_{-a} \rightarrow \tau_{-a}}(\ell)$ for all $\ell \in [n-2]$ and $f_{\A}^{\sigma \rightarrow \tau}(n-1) = 0$. Note that alternative $a$ continues to be at the $k^\text{th}$ position in the current list at the end of the first $(k-1)$ steps.

At the $k^\text{th}$ step, $\A$ performs an insertion sort operation. Since $a$ is the most preferred alternative under $\tau$, this step once again involves promoting $a$ by $(k-1)$ positions to the top of the current list, i.e., the count function is modified exactly as in Case I. The rest of the analysis is identical to Case I as well. This finishes the proof of \Cref{lem:EquivalenceOfSortingAlgorithms}.
\end{proof}

\begin{restatable}{lemma}{SameTopAlternative}	
	Let $a \in A$ and $\sigma_{-a},\tau_{-a} \in \ml(A \setminus \{a\})$. Let $\sigma,\tau \in \ml(A)$ be such that $\sigma \coloneqq (a,\sigma_{-a})$ and $\tau \coloneqq (a,\tau_{-a})$. Then, for any sorting algorithm $\A \in \mathfrak{A}$, $f_{\A}^{\sigma \rightarrow \tau}(\ell) = f_{\A}^{\sigma_{-a} \rightarrow \tau_{-a}}(\ell)$ for all $\ell \in [m-2]$ and $f_{\A}^{\sigma \rightarrow \tau}(m-1) = 0$, where $|A| = m$.
	\label{lem:SameTopAlternative}
\end{restatable}
\begin{proof}
We will first argue that $f_{\A}^{\sigma \rightarrow \tau}(m-1) = 0$. Suppose, for contradiction, that $f_{\A}^{\sigma \rightarrow \tau}(m-1) > 0$, that is, some alternative (say, $b$) is promoted by $(m-1)$ positions during the execution of $\A$. Since both selection and insertion sort maintain the sorted prefix property at every time step, it must be that $b \>_{\tau} a$, which is a contradiction since $a$ is the most preferred alternative under $\tau$.

Next, we will argue that $f_{\A}^{\sigma \rightarrow \tau}(\ell) = f_{\A}^{\sigma_{-a} \rightarrow \tau_{-a}}(\ell)$ for all $\ell \in [m-2]$. Once again, by the sorted prefix property, no alternative is promoted above $a$ at any time step during $\sigma \rightarrow \tau$. Since the top position remains fixed, the execution of $\A$ during $\sigma_{-a} \rightarrow \tau_{-a}$ can be mimicked to obtain the execution of $\A$ during $\sigma \rightarrow \tau$. The lemma now follows.
\end{proof}

\subsection{Proof of Theorem~\ref{thm:LinearWeightKTdist}}
\label{subsec:Proof_LinearWeightKTdist}

\LinearWeightKTdist*
\begin{proof}
	For the linear weight function $\wlin{}$, we have $\cost_{\wlin}(\sigma,\tau) = \sum_{\ell = 1}^{m-1} f^{\sigma \rightarrow \tau}(\ell) \cdot \ell$. Regardless of the choice of the sorting algorithm, any fixed pair of alternatives is swapped at most once during the transformation from $\sigma$ to $\tau$. As a result, each ``move up by $\ell$ slots'' operation, which contributes $\ell$ units to the $\cost$ function, also contributes $\ell$ units to the Kendall's Tau distance, giving us the desired claim. 
	
	For the affine weight function $\waff{}$, we therefore have 
	\begin{align*}
	\cost_{\waff}(\sigma,\tau)
	& = c \cdot \cost_{\wlin}(\sigma,\tau) + d \cdot \sum_{\ell = 1}^{m-1} f^{\sigma \rightarrow \tau}(\ell).\\
	& = c \cdot \dkt{}(\sigma,\tau) + d \cdot \#\textup{moves},
	\end{align*}
	as desired.
\end{proof}

\subsection{Proof of Theorem~\ref{thm:ExactAlgorithms}}
\label{subsec:Proof_ExactAlgorithms}

\ExactAlgorithms*
\begin{proof}
	(a) \emph{When $\D$ is a $k$-mixture Plackett-Luce $(\kPL)$ with $k = 1$}\\
	
	The expected time for any $\sigma \in \ml(A)$ under the \PL{} model with the parameter $\mytheta$ is given by
	\begin{align}
	\E_{\tau \sim \mytheta} [\cost_{\wlin}(\sigma,\tau)] \nonumber & = \E_{\tau \sim \mytheta} [\dkt(\sigma,\tau)] &  \text{(\Cref{thm:LinearWeightKTdist})} \nonumber \\
	& = \E_{\tau \sim \mytheta} \left[ \sum_{a_i,a_j \in A \, : \, a_i \>_{\sigma} a_j} \Ind[a_j \succ_{\tau} a_i] \right]  &  \text{(\Cref{defn:KendallTau})} \nonumber \\
	& = \sum_{a_i,a_j \in A \, : \, a_i \>_{\sigma} a_j} \E_{\tau \sim \mytheta} \left[ \Ind[a_j \succ_{\tau} a_i] \right] &  \text{(Linearity of Expectation)} \nonumber \\
	& = \sum_{a_i,a_j \in A \, : \, a_i \>_{\sigma} a_j} \Pr_{\tau \sim \mytheta} ( a_j \succ_{\tau} a_i ) &  \nonumber \\
	& = \sum_{a_i,a_j \in A \, : \, a_i \>_{\sigma} a_j} \frac{\theta_j}{\theta_i + \theta_j} & \text{(\Cref{defn:PL})}.
	\label{eqn:PL-KT-polytime-temp1}
	\end{align}
	Let $\sigma^* \in \ml(A)$ be a linear order that is consistent with the parameter $\mytheta$. That is, for any $a_i,a_j \in A$, $a_i \>_{\sigma^*} a_j$ if and only if either $\theta_i > \theta_j$ or $i < j$ in case $\theta_i = \theta_j$. We will show via an exchange argument that for any $\sigma \in \ml(A)$, $\E_{\tau \sim \mytheta} [\cost_{\wlin}(\sigma^*,\tau)] \leq \E_{\tau \sim \mytheta} [\cost_{\wlin}(\sigma,\tau)]$. The desired implication will then follow by simply computing $\sigma^*$, which can be done in polynomial time.
	
	Consider a pair of alternatives $a_i,a_j \in A$ that are adjacent in $\sigma$ such that $a_i \>_{\sigma^*} a_j$ and $a_j \>_{\sigma} a_i$ (such a pair must exist as long as $\sigma \neq \sigma^*$).\footnote{Note that $a_i,a_j$ need not be adjacent in $\sigma^*$.} Let $\sigma' \in \ml(A)$ be derived from $\sigma$ by swapping $a_i$ and $a_j$ (and making no other changes). Then, from \Cref{eqn:PL-KT-polytime-temp1}, we have that
	\begin{align*}
	\E_{\tau \sim \mytheta} [\cost_{\wlin}(\sigma',\tau)] - \E_{\tau \sim \mytheta} & [\cost_{\wlin}(\sigma,\tau)] = \frac{\theta_j}{\theta_i + \theta_j} - \frac{\theta_i}{\theta_i + \theta_j} \leq 0,
	\end{align*}
	where the inequality holds because $\sigma^*$ is consistent with $\mytheta$ and $a_i \>_{\sigma^*} a_j$. By repeated use of the above argument---with $\sigma'$ taking the role of $\sigma$, and so on---we get the desired claim.\\
	
	(b) \emph{When $\D$ is a $k$-mixture Mallows model $(\kMM{})$ with $k = 1$}\\
	
	The proof is similar to case (a). Once again, we let $\sigma$ and $\sigma'$ be two linear orders that are identical except for the pair $a_i,a_j \in A$ that are adjacent in $\sigma$ such that $a_i \>_{\sigma^*} a_j$, $a_j \>_{\sigma} a_i$, and $a_i \>_{\sigma'} a_j$; here $\sigma^*$ is the reference ranking for the Mallows model. Then,
	\begin{align*}
	\E_{\tau \sim (\sigma^*,\phi)} [\cost_{\wlin}(\sigma',\tau)] - & \E_{\tau \sim (\sigma^*,\phi)} [\cost_{\wlin}(\sigma,\tau)]& \\
	& = \E_{\tau \sim (\sigma^*,\phi)} [\dkt(\sigma',\tau)] - \E_{\tau \sim (\sigma^*,\phi)} [\dkt(\sigma,\tau)] & \text{(by \Cref{thm:LinearWeightKTdist})}\\
	& = \Pr_{\tau \sim (\sigma^*,\phi)} ( a_j \succ_{\tau} a_i ) - \Pr_{\tau \sim (\sigma^*,\phi)} ( a_i \succ_{\tau} a_j ) &\\
	& = 2 \left( \frac{1}{2} - \Pr_{\tau \sim (\sigma^*,\phi)} ( a_i \succ_{\tau} a_j ) \right) &\\
	& = 2 \left( \frac{1}{2} - \frac{\sum_{z=1}^{\Delta} z \phi^{z-1}}{\left( \sum_{z=0}^{\Delta-1} \phi^{z} \right) \left( \sum_{z=0}^{\Delta} \phi^{z} \right)}  \right) & \text{(by \Cref{prop:MM-pairwise})},
	\end{align*}
	where $\Delta = \rank(\sigma^*,a_j) - \rank(\sigma^*,a_i)$. It is easy to verify that 
	$g(\Delta) \coloneqq \frac{\sum_{z=1}^{\Delta} z \phi^{z-1}}{\left( \sum_{z=0}^{\Delta-1} \phi^{z} \right) \left( \sum_{z=0}^{\Delta} \phi^{z} \right)} \geq \frac{1}{2}$
	for all integral $\Delta \geq 1$ whenever $\phi \in [0,1]$. This implies that
	\begin{align*}
	    \E_{\tau \sim (\sigma^*,\phi)} [\cost_{\wlin}(\sigma',\tau)] \leq \E_{\tau \sim (\sigma^*,\phi)} [\cost_{\wlin}(\sigma,\tau)].
	\end{align*}
	
	Repeated application of the above argument shows that for any linear order $\sigma \in \ml(A)$,
	\begin{align*}
	    \E_{\tau \sim (\sigma^*,\phi)} [\cost_{\wlin}(\sigma^*,\tau)] \leq \E_{\tau \sim (\sigma^*,\phi)} [\cost_{\wlin}(\sigma,\tau)].
	\end{align*}
	
	The desired implication follows by simply returning the reference ranking $\sigma^*$ as the output.\\
	
	(c) \emph{When $\D$ is a uniform distribution with support size $n \leq 2$}\\
	
	Let $\D$ be a uniform distribution over the set of $n$ linear orders $\{\sigma^{(i)}\}_{i=1}^n$. From \Cref{thm:LinearWeightKTdist}, we know that for any $\sigma \in \ml(A)$, we have $\E_{\tau \sim \D} [\cost_{\wlin}(\sigma,\tau)] =\sum_{i=1}^n \dkt(\sigma,\sigma^{(i)})$. When $n=1$, it is clear that $\sigma = \sigma^{(1)}$ is the unique minimizer of the expected cost. When $n=2$, it can be argued that $\sigma \in \{\sigma^{(1)},\sigma^{(2)}\}$ is the desired solution. Indeed, let $S \coloneqq \{(a_i,a_j) \in A \times A : a_i \>_{\sigma^{(1)}} a_j \text{ and } a_j \>_{\sigma^{(2)}} a_i\}$ be the set of (ordered) pairs of alternatives over which $\sigma^{(1)}$ and $\sigma^{(2)}$ disagree. Any linear order $\sigma \notin \{\sigma^{(1)},\sigma^{(2)}\}$ contributes at least $|S|$ to the expected time in addition to the number of pairs over which $\sigma$ differs from $\sigma^{(1)}$ or $\sigma^{(2)}$. Hence, the expected time is minimized when $\sigma \in \{\sigma^{(1)},\sigma^{(2)}\}$.
\end{proof}

\subsection{Proof of Theorem~\ref{thm:HardnessResults}}
\label{subsec:Proof_HardnessResults}

\HardnessResults*
\begin{proof}
	(a) \emph{When $\D$ is $k$-mixture Plackett-Luce model $(\kPL{})$ for $k=4$}\\
	
	Let $\D$ be a $k$-mixture Plackett-Luce model with the parameters $\{\gamma^{(\ell)},\mytheta^{(\ell)}\}_{\ell=1}^k$, and let $\sigma \in \ml(A)$. By an argument similar to that in the proof of \Cref{thm:ExactAlgorithms}, we have that
	\begin{align}
	\E_{\tau \sim \kPL} & [\cost_{\wlin}(\sigma,\tau)] = \sum_{a_i,a_j \in A \, : \, a_i \>_{\sigma} a_j} \sum_{\ell = 1}^k \gamma^{(\ell)} \cdot \frac{\theta^{(\ell)}_j}{\theta^{(\ell)}_i + \theta^{(\ell)}_j},
	\label{eqn:kPL_NPC_temp1}
	\end{align}
	which can be computed in polynomial time. Hence the problem is in \NP{}.
	
	To prove \NPH{}ness{}, we will show a reduction from a restricted version of \Kem{} for four agents, which is known to be \NPC{} \citep{Dwork01:Rank}. Given any instance of \Kem{} with the preference profile $\{\sigma^{(\ell)}\}_{\ell=1}^n$ where $n=4$, the parameters of $\D$ are set up as follows: The number of mixtures is set to $k = n = 4$. For each $\ell \in [n]$, $\gamma^{(\ell)} = \frac{1}{n}$, and for each $a_i \in A$, $\theta_i^{(\ell)} = m^{4(m - \rank(\sigma^{(\ell)},a_i))}$. Thus, for instance, if $\sigma^{(1)} = (a_1 \> a_2 \> \dots a_m)$, then $\theta^{(1)}_1 = m^{4(m-1)}, \theta^{(1)}_2 = m^{4(m-2)}, \dots, \theta^{(1)}_m = 1$. Notice that despite being exponential in $m$, the parameters $\{\theta_i^{(\ell)}\}_{\ell \in [n], i \in [m]}$ can each be specified in $\poly(m)$ number of bits, and are therefore polynomial in the input size.
	
	We will now argue that a linear order $\sigma \in \ml(A)$ satisfies  $\sum_{i=1}^n \dkt(\sigma,\sigma^{(i)}) \leq \delta$ if and only if $\E_{\tau \sim \kPL} [\cost_{\wlin}(\sigma,\tau)] \leq \delta + 0.5$. First, suppose that $\sigma$ satisfies  $\sum_{i=1}^n \dkt(\sigma,\sigma^{(i)}) \leq \theta$. Define, for each $\ell \in [n]$, $S_\ell \coloneqq \{ (a_i,a_j) \in A \times A : a_i \>_\sigma a_j \text{ and } a_j \>_{\sigma^{(\ell)}} a_i \}$. Thus, $\sum_{\ell = 1}^n |S_\ell| \leq \delta$. Then, from \Cref{eqn:kPL_NPC_temp1}, we have that
	\begin{align*}
	\E_{\tau \sim \kPL} [\cost_{\wlin}(\sigma,\tau)] & = \sum_{a_i,a_j \in A \, : \, a_i \>_{\sigma} a_j} \sum_{\ell = 1}^n \frac{1}{n} \cdot \frac{\theta^{(\ell)}_j}{\theta^{(\ell)}_i + \theta^{(\ell)}_j}\\
	& = \frac{1}{n} \sum_{\ell = 1}^n \left( \sum_{(a_i,a_j) \in S_\ell} \frac{\theta^{(\ell)}_j}{\theta^{(\ell)}_i + \theta^{(\ell)}_j} + \sum_{(a_i,a_j) \notin S_\ell} \frac{\theta^{(\ell)}_j}{\theta^{(\ell)}_i + \theta^{(\ell)}_j} \right)\\
	& \leq \frac{1}{n} \sum_{\ell = 1}^n \left( \sum_{(a_i,a_j) \in S_\ell} 1 + \sum_{(a_i,a_j) \notin S_\ell} \frac{1}{1 + m^4} \right)\\
	& \leq \frac{1}{n} \sum_{\ell = 1}^n \left( \delta + \binom{m}{2} \frac{1}{1 + m^4} \right)\\
	& = \delta + \binom{m}{2} \frac{1}{1 + m^4}\\
	& \leq \delta + 0.5.
	\end{align*}
	The first inequality follows from the choice of parameters in our construction. Indeed, for any $(a_i,a_j) \in S_\ell$, we have $\theta^{(\ell)}_i, \theta^{(\ell)}_j \geq 1$ and therefore $\frac{\theta^{(\ell)}_j}{\theta^{(\ell)}_i + \theta^{(\ell)}_j} \leq 1$. In addition, for any $(a_i,a_j) \notin S_\ell$, we have that $\rank(\sigma^{(\ell)},a_i) < \rank(\sigma^{(\ell)},a_j)$, and therefore
	\begin{align*}
	\frac{\theta^{(\ell)}_j}{\theta^{(\ell)}_i + \theta^{(\ell)}_j} =  \frac{ 1 }{1 + m^{4 \left( \rank(\sigma^{(\ell)},a_j) - \rank(\sigma^{(\ell)},a_i) \right) } } \leq \frac{1}{1+m^4}.
	\end{align*}
	
	The second inequality uses the fact that $\sum_{\ell = 1}^n |S_\ell| \leq \delta$. The final inequality holds because $m^4 + 1 - 2 \binom{m}{2} > 0$ for all $m \geq 1$.
	
	Now suppose that $\sigma$ satisfies $\E_{\tau \sim \kPL} [\cost_{\wlin}(\sigma,\tau)] \leq \delta + 0.5$. We will argue that $\sum_{i=1}^n \dkt(\sigma,\sigma^{(i)})$ must be strictly smaller than $\delta + 1$, which, by integrality, will give us the desired claim. Suppose, for contradiction, that $\sum_{i=1}^n \dkt(\sigma,\sigma^{(i)}) \geq \delta + 1$. Thus, $\sum_{\ell = 1}^n |S_\ell| \geq \delta + 1$. We can use this relation to construct a lower bound on the expected time, as follows:
	\begin{align*}
	\E_{\tau \sim \kPL}[\cost_{\wlin}(\sigma,\tau)] & = \frac{1}{n} \sum_{\ell = 1}^n \left( \sum_{(a_i,a_j) \in S_\ell} \frac{\theta^{(\ell)}_j}{\theta^{(\ell)}_i + \theta^{(\ell)}_j} + \sum_{(a_i,a_j) \notin S_\ell} \frac{\theta^{(\ell)}_j}{\theta^{(\ell)}_i + \theta^{(\ell)}_j}\right) \\
	& \geq (\delta + 1) \frac{m^4}{m^4 + 1}\\
	& > \delta + 0.5.
	\end{align*}
	The first inequality holds because for any $(a_i,a_j) \notin S_\ell$, we have that $\frac{\theta^{(\ell)}_j}{\theta^{(\ell)}_i + \theta^{(\ell)}_j} \geq 0$, and for any $(a_i,a_j) \in S_\ell$, we have that
	\begin{align*}
	    \frac{\theta^{(\ell)}_j}{\theta^{(\ell)}_i + \theta^{(\ell)}_j} \geq \frac{m^{4 ( \rank(\sigma^{(\ell)},a_i) - \rank(\sigma^{(\ell)},a_j) ) }}{1+m^{4 ( \rank(\sigma^{(\ell)},a_i) - \rank(\sigma^{(\ell)},a_j) ) }} \geq \frac{m^4}{1+m^4}.
	\end{align*}

	 The second inequality holds because $\frac{\delta+1}{1+m^4} < \frac{1}{2}$ for all $m \geq 2$ (note that we can assume $m \geq 2$ without loss of generality). The chain of inequalities give us the desired contradiction. Hence, it must be that $\sum_{i=1}^n \dkt(\sigma,\sigma^{(i)}) \leq \delta$. This finishes the proof of part (a) of \Cref{thm:HardnessResults}.\\

	(b) \emph{When $\D$ is $k$-mixture Mallows model $(\kMM{})$ for $k=4$}\\
	
	Let $\D$ be a $k$-mixture Mallows model with the parameters $\{\gamma^{(\ell)},\sigma^*_{(\ell)},\phi_{(\ell)}\}_{\ell=1}^k$, and let $\sigma \in \ml(A)$. By an argument similar to that in the proof of \Cref{thm:ExactAlgorithms}, we have that
	\begin{align*}
	    & \E_{\tau \sim \kMM} [\cost_{\wlin}(\sigma,\tau)] = \sum_{a_i,a_j \in A \, : \, a_i \>_{\sigma} a_j} \sum_{\ell=1}^k \gamma^{(\ell)} \cdot \Pr_{\tau \sim (\sigma^*_{(\ell)},\phi_{(\ell)})} ( a_j \succ_{\tau} a_i ),
	\end{align*}
	which can be computed in polynomial time (\Cref{eqn:PairwiseMarginalkMM}). Therefore, the problem is in \NP{}.
	
	To prove \NPH{}ness, we will show a reduction from \Kem{}. Given any instance of \Kem{} with the preference profile $\{\sigma^{(\ell)}\}_{\ell=1}^n$, the parameters of $\D$ are set up as follows: The number of mixtures $k$ is set to $n$. For each $\ell \in [n]$, $\gamma_{(\ell)} = \frac{1}{n}$, $\sigma^*_{(\ell)} = \sigma^{(\ell)}$, and $\phi_{(\ell)} = 0$. The expected time for any linear order $\sigma$ is simply its average Kendall's Tau distance from the profile $\{\sigma^{(\ell)}\}_{\ell=1}^n$, hence the equivalence of the solutions follows. Finally, since \Kem{} is known to be \NPC{} even for $n = 4$, a similar implication holds for \Rec{} when $k = 4$.\\
	
	(c) \emph{When $\D$ is a uniform distribution over $n=4$ linear orders}\\
	
	Membership in \NP{} follows from \Cref{thm:LinearWeightKTdist}, since for the linear weight function, the expected time for any linear order $\sigma \in \ml(A)$ is equal to its average Kendall's Tau distance from the preference profile that supports $\D$, which can be computed in polynomial time. In addition, \NPH{}ness follows from a straightforward reduction from \Kem{}: Given any instance of \Kem{} with the preference profile $\{\sigma^{(i)}\}_{i=1}^n$, the distribution $\D$ in \Rec{} is simply a uniform distribution over $\{\sigma^{(i)}\}_{i=1}^n$. The equivalence of the solutions follows once again from \Cref{thm:LinearWeightKTdist}. Finally, since \Kem{} is known to be \NPC{} even for $n = 4$, a similar implication holds for \Rec{} as well.
\end{proof}

\subsection{Proof of Theorem~\ref{thm:RecPTAS}}
\label{subsec:Proof_RecPTAS}

\RecPTAS*
\begin{proof}
	We will show that for each of the three settings in \Cref{thm:RecPTAS}, \Rec{} turns out to be a special case of \wfast{}, and therefore the \PTAS{} of \Cref{prop:KendallTauPTAS} from \Cref{subsec:Computational_Problems} applies.\\
	
	(a) \emph{When $\D$ is $k$-mixture Plackett-Luce model $(\kPL{})$}\\
	
	Recall from \Cref{thm:LinearWeightKTdist} that when the weight function is linear, the expected cost of $\sigma \in \ml(A)$ is given by $\E_{\tau \sim \D} [\cost_{\wlin}(\sigma,\tau)] = \E_{\tau \sim \D} [\dkt(\sigma,\tau)]$. When $\D$ is a $k$-mixture Plackett-Luce model (\kPL{}) with the parameters $\{\gamma^{(\ell)},\theta^{(\ell)}\}_{\ell = 1}^k$, the expected cost of $\sigma$ under $\D$ is given by (refer to \Cref{eqn:PL-KT-polytime-temp1} in the proof of \Cref{thm:ExactAlgorithms}):
	\begin{align*}
	    \E_{\tau \sim \kPL}& [\cost_{\wlin}(\sigma,\tau)] = \sum_{a_i,a_j \in A \, : \, a_i \>_{\sigma} a_j} \sum_{\ell = 1}^k \gamma^{(\ell)} \cdot \frac{\theta^{(\ell)}_j}{\theta^{(\ell)}_i + \theta^{(\ell)}_j}.
	\end{align*}
	
	Consider a complete, directed, and weighted graph $G = (A,E)$ defined over the set of alternatives, where for every pair of alternatives $a_i,a_j$, we have $(a_i,a_j) \in E$ if and only if either $\theta_i > \theta_j$ or $i < j$ in case $\theta_i = \theta_j$. Each edge $(a_i,a_j) \in E$ is associated with a pair of weights $w_{i,j} = \sum_{\ell = 1}^k \gamma^{(\ell)} \cdot \frac{\theta^{(\ell)}_i}{\theta^{(\ell)}_i + \theta^{(\ell)}_j}$ and $w_{j,i} = \sum_{\ell = 1}^k \gamma^{(\ell)} \cdot \frac{\theta^{(\ell)}_j}{\theta^{(\ell)}_i + \theta^{(\ell)}_j}$. Notice that $w_{i,j} + w_{j,i} = 1$ for every $(a_i,a_j) \in E$. Furthermore, the expected cost of $\sigma$ can be expressed in terms of the edge-weights as follows:
	$$\E_{\tau \sim \kPL} [\cost_{\wlin}(\sigma,\tau)] = \sum_{a_i,a_j \in A} w_{j,i} \cdot \Ind[a_i \>_{\sigma} a_j].$$
	Therefore, $\sigma$ is a solution of \Rec{} if and only if it is a solution of \wfast{} for the graph $G$ constructed above (with $b = 1$).\\
	
	(b) \emph{When $\D$ is $k$-mixture Mallows model $(\kMM{})$}\\
	
	An analogous argument works for the case when $\D$ is a $k$-mixture Mallows model (\kMM{}) with the parameters $\{\gamma^{(\ell)},\sigma^*_{(\ell)},\phi_{(\ell)}\}_{\ell=1}^k$. In this case, we set the weights to be $w_{i,j} = \sum_{\ell=1}^{k} \gamma^{(\ell)} \cdot g_{\ell}(\Delta^{i,j}_\ell) \text{ and } w_{j,i} = \sum_{\ell=1}^{k} \gamma^{(\ell)} \cdot \left( 1 -  g_{\ell}(\Delta^{i,j}_\ell) \right)$, where $\Delta^{i,j}_{\ell}$ and $g_{\ell}(\cdot)$ are as defined in \Cref{eqn:PairwiseMarginalkMM} in \Cref{subsec:Appendix_Preliminaries}.\\
	
	(c) \emph{When $\D$ is a uniform distribution}\\
	
	Finally, when $\D$ is a uniform distribution over $\{\sigma^{(\ell)}\}_{\ell=1}^n$, an analogous argument works for $w_{i,j} = \sum_{\ell=1}^{n} \frac{1}{n} \cdot \Ind[a_i \>_{\sigma^{(\ell)}} a_j]$ and $w_{j,i} = \sum_{\ell=1}^{n} \frac{1}{n} \cdot \Ind[a_j \>_{\sigma^{(\ell)}} a_i]$.
\end{proof}

\subsection{Proof of Theorem~\ref{thm:BordaApproxAlgo}}
\label{subsec:Proof_BordaApproxAlgo}

\BordaApproxAlgo*
\begin{proof} (Sketch) The proof is similar to that of \Cref{thm:RecPTAS} in \Cref{subsec:Proof_RecPTAS}. The only difference is that we use the algorithm in \Cref{prop:BordaApprox} (from \Cref{subsec:Computational_Problems}) instead of \Cref{prop:KendallTauPTAS} as a subroutine. Notice that the condition $w_{i,j} + w_{j,i} = 1$ is satisfied for every $(a_i,a_j) \in E$, and thus \Cref{prop:BordaApprox} is applicable.
\end{proof}

\subsection{Proof of Theorem~\ref{thm:ApproxGeneralWeights}}
\label{subsec:Proof_ApproxGeneralWeights}

\ApproxGeneralWeights*
\begin{proof}
	We will show that the linear order $\sigma$ constructed in \Cref{thm:RecPTAS} provides the desired approximation guarantee. Let $\sigma^\lin \in \arg\min_{\sigma \in \ml(A)} \E_{\tau \sim \D} [\cost_{\wlin}(\sigma,\tau)]$. Then,
	\begin{align*}
	\E_{\tau \sim \D} [\cost_{\w}(\sigma,\tau)] & \leq \alpha \E_{\tau \sim \D} [\cost_{\wlin}(\sigma,\tau)] & \quad \text{(by closeness-of-weights)} \\
	& \leq \alpha (1+\varepsilon) \E_{\tau \sim \D} [\cost_{\wlin}(\sigma^\lin,\tau)] & \quad \text{(\Cref{thm:RecPTAS})} \\
	& \leq \alpha (1+\varepsilon) \E_{\tau \sim \D} [\cost_{\wlin}(\sigma^*,\tau)] & \quad \text{(optimality of $\sigma^\lin$)} \\
	& \leq \alpha \beta (1+\varepsilon) \E_{\tau \sim \D} [\cost_{\w}(\sigma^*,\tau)] & \quad \text{(by closeness-of-weights)}.
	\end{align*}
\end{proof}

\end{document}